%% file: main.tex
\documentclass{article}  

\newif\ifpubrel
\pubrelfalse             

\author{Daniel J.  Dougherty$^1$ \and Joshua D. Guttman$^{1,2}$}
\title{Molly: A Verified Compiler for Cryptoprotocol Roles}
\date{%
  $^1$Worcester Polytechnic Institute \\%
  $^2$The MITRE Corporation %
}

\input{macros}

\input{macros_this_paper}

\notabenetrue

\fancypagestyle{title}{%
  \fancyhf{}%
  \fancyhead[L]{Approved for Public Release; Distribution
    Unlimited. Case Number 23-3747.}%
  \fancyfoot[C]{\footnotesize This technical data was developed using
    contract funds under Basic Contract No.  W56KGU-22-F-0017. The
    view, opinions, and/or findings contained in this report are those
    of The MITRE Corporation and should not be construed as an
    official Government position, policy, or decision, unless
    designated by other documentation.  Copyright {(c)} 2023 The MITRE
    Corporation. ALL RIGHTS RESERVED.}%
}

\begin{document}
\maketitle
\ifpubrel
  \thispagestyle{title}
\fi
\begin{abstract}
  Molly is a program that compiles cryptographic protocol roles
  written in a high-level notation into straight-line programs in an
  intermediate-level imperative language, suitable for implementation
  in a conventional programming language.

  We define a denotational semantics for protocol roles based on an
  axiomatization of the runtime.  A notable feature of our approach is
  that we assume that encryption is randomized.  Thus, at the runtime
  level we treat encryption as a relation rather than a function.

  Molly is written in Coq, and generates a machine-checked proof that
  the procedure it constructs is correct with respect to the runtime
  semantics.  Using Coq's extraction mechanism, one can build an
  efficient functional program for compilation.
\end{abstract}

\newpage
\tableofcontents
\newpage

\input{body}

\subsection*{Acknowledgments.}  We are grateful  to John Ramsdell for
several helpful discussions.

\bibliographystyle{alpha}
\bibliography{djd,cpsa}

\end{document}

%% file: macros.tex
\newif\ifllncs
\llncsfalse  

\usepackage{url}
\usepackage{amssymb}
\usepackage{amsmath}
\usepackage{amsthm}

\allowdisplaybreaks

\usepackage[all]{xy}
\usepackage{xcolor}

\usepackage{stmaryrd}

\usepackage{pdfsync} 

\newif\ifproofs

\newif\ifnotabene

\newcommand{\cpsa}{\textsc{cpsa}}

\newcommand{\option}[1]{\ensuremath {{#1}}_{\bot}}
\newcommand{\pow}[1]{\ensuremath 2^{#1}}

\newcommand{\inv}[1]{\ensuremath{ {#1}^{-1} }}

\newcommand{\nat}{\mathbb{N}}

\newcommand{\code}[1]{\ensuremath{\ulcorner{#1}\urcorner}}

\usepackage{fancyhdr} 
\usepackage{comment}
\usepackage{xspace}
\usepackage{listings}
\usepackage{enumitem}
\usepackage{multicol}

\usepackage[xcolor,outerbars]{changebar}  
\cbcolor{red}
\usepackage{tikz-cd}

\usepackage{hyperref}
\hypersetup{colorlinks=true,citecolor=red}

\theoremstyle{plain}   
\newtheorem{theorem}{Theorem}  [section]
\newtheorem*{theorem*}{Theorem}
\newtheorem{lemma}[theorem]{Lemma}
\newtheorem{lemma*}{Lemma}
\newtheorem{proposition}[theorem]{Proposition}
\newtheorem*{proposition*}{Proposition}

\newtheorem*{corollary*}{Corollary}

\theoremstyle{definition}  
\newtheorem{definition}[theorem]{Definition}
\newtheorem*{definition*}{Definition}
\newtheorem{notation}[theorem]{Notation}
\newtheorem*{notation*}{Notation}
\newtheorem{remark}[theorem]{Remark}
\newtheorem*{remark*}{Remark}

\newtheorem*{note*}{Note}

\newtheorem*{notes*}{Notes}

\newtheorem*{observation*}{Observation}

\newtheorem*{fact*}{Fact}

\newtheorem*{facts*}{Facts}

\newtheorem*{caution*}{Caution!}

\newtheorem{example}{Example}
\newtheorem*{example*}{Example}

\newtheorem*{examples*}{Examples}

\newcount\timehh\newcount\timemm \timehh=\time
\divide\timehh by 60 \timemm=\time
\newcommand{\timestamp}{
  {\protect\small\sl\today\ --
    \ifnum\timehh<10 0\fi\number\timehh\,:\,
    \ifnum\timemm<10 0\fi\number\timemm}}

\def\refsec#1{Section~\ref{#1}}
\def\refdef#1{Definition~\ref{#1}}
\def\reflem#1{Lemma~\ref{#1}}
\def\refprop#1{Proposition~\ref{#1}}
\def\refthm#1{Theorem~\ref{#1}}

\def\refeg#1{Example~\ref{#1}}

\newcommand{\eqdef}{\overset{\mbox{\tiny def}}{=}}  

\newcommand{\ie}{\textit{i.e.}\xspace}

\newcommand{\eg}{\textit{e.g.}\xspace}

\newcounter{running}[section]
  {\setcounter{running}{\value{enumi}}\end{enumerate}}

\newcounter{runningtwo}[section]
  {\setcounter{runningtwo}{\value{enumii}}\end{enumerate}}

\usepackage{bussproofs}

\newlength{\twoheadspace}  
\newlength{\backupamount}  
\newlength{\lrbackupamount} 
\usepackage{xcolor}

\definecolor{dkblue}{rgb}{0,0.1,0.6}

\definecolor{lightblue}{rgb}{0,0.5,0.5}
\definecolor{ltblue}{rgb}{0,0.4,0.4}

\definecolor{dkgreen}{rgb}{0,0.35,0}
\definecolor{dk2green}{rgb}{0.4,0,0}

\definecolor{dkviolet}{rgb}{0.3,0,0.5}

\definecolor{dkred}{rgb}{0.5,0,0}
\definecolor{bggray}{gray}{0.95}

\usepackage{sourcecodepro}

\usepackage[T1]{fontenc}


%% file: macros_this_paper.tex
\let\Pr\undefined

\newcommand{\AnB}{Alice \& Bob}

\newcommand{\molly}{\textsc{Molly}\xspace}

\newcommand{\pub}[1]{\ensuremath{({pub}\; {#1} )}\xspace}
\newcommand{\pri}[1]{\ensuremath{({pri}\; {#1} )}\xspace}

\newcommand{\newsf}[2]{
    \newcommand{#1}{\ensuremath{\mathop{\mathsf{#2}}}\xspace}}

 \newcommand{\newbf}[2]{
   \newcommand{#1}{\ensuremath{\mathop{\mathbf{#2}}}\xspace}}

 \newcommand{\newtt}[2]{
   \newcommand{#1}{\ensuremath{\mathop{\mathtt{#2}}}\xspace}}

 \newcommand{\newrm}[2]{
   \newcommand{#1}{\ensuremath{\mathop{\mathrm{#2}}}\xspace}}

\newsf{\Sort}{Sort}

\newbf{\Chan}{chan}
\newbf{\Text}{text}
\newbf{\Data}{data}
\newbf{\Name}{name}
\newbf{\Skey}{skey}
\newbf{\Akey}{akey}
\newbf{\Ikey}{ikey} 

\newbf{\Mesg}{mesg}

\newrm{\strings}{string}
\newrm{\bool}{bool}
\newrm{\skeys}{skeys}
\newrm{\akeys}{akeys}

\newcommand{\newact}{\newsf}

\newact{\Actop}{Act}
\newcommand{\Act}[1]{\ensuremath({\Actop}\;{#1})}

\newact{\Prmop}{Prm}
\newcommand{\Prm}[1]{\ensuremath({\Prmop}\;{#1})}

\newact{\Retop}{Ret}
\newcommand{\Ret}[1]{\ensuremath({\Retop}\;{#1})}

\newact{\Rcvop}{Rcv}
\newcommand{\Rcv}[2]{\ensuremath({\Rcvop}\;{#1}\;{#2})}

\newact{\Sndop}{Snd}
\newcommand{\Snd}[2]{\ensuremath({\Sndop}\;{#1}\;{#2})}

\newcommand{\actmap}[1]{{#1}^{\Actop}}

\newcommand{\rtThy}{\ensuremath{\mathcal{R}}\xspace}

\newcommand{\newrt}{\newtt}

\newrm{\rtval}{Rtval}

\newrt{\rtsrt}{rtsort}

\newrt{\rtpair}{pair}
\newrt{\rtfrst}{frst}
\newrt{\rtscnd}{scnd}
\newrt{\rtencr}{encr}
\newrt{\rtdecr}{decr}
\newrt{\rthash}{hash}
\newrt{\rtquote}{quot}
\newrt{\rtgen}{gen}

\newrt{\rtpubof}{pubof}

\newrt{\rtsame}{same}
\newrt{\rtchcksrt}{checksrt}
\newrt{\rtchckhash}{checkhash}
\newrt{\rtchckquote}{checkquote}

\newrt{\rtkypr}{kypr}

\newrt{\rtinvop}{rtinv}
 \newcommand{\rtinv}[2]{\ensuremath({\rtinvop}\;{#1}\;{#2})}

\newrt{\true}{true}
\newrt{\false}{false}

\newcommand{\opteq}{\shortdownarrow}

\newsf{\tr}{tr}

\newrm{\Term}{Term}
\newrm{\Var}{var}

\newrm{\role}{role}
\newrm{\roles}{roles}
\newsf{\rl}{rl}

\newrm{\sortof}{sort}

\newcommand{\newtrm}{\newsf}

\newtrm{\Enop}{En}
\newcommand{\En}[2]{\ensuremath({\Enop}\;{#1}\;{#2})}

\newsf{\Prrop}{Pr}
\newcommand{\Prr}[2]{\ensuremath{(\Prrop}\;{#1}\;{#2})}
\newcommand{\Pr}{\Prr}

\newtrm{\Hsop}{Hs}
\newcommand{\Hs}[1]{\ensuremath({\Hsop}\;{#1})}

\newtrm{\Qtop}{Qt}
\newcommand{\Qt}[1]{\ensuremath({\Qtop}\;{#1})}

\newtrm{\Chop}{Ch}
\newcommand{\Ch}[1]{\ensuremath({\Chop}\;{#1})}

\newtrm{\Txop}{Tx}
\newcommand{\Tx}[1]{\ensuremath({\Txop}\;{#1})}

\newtrm{\Dtop}{Dt}
\newcommand{\Dt}[1]{\ensuremath({\Dtop}\;{#1})}

\newtrm{\Nmop}{Nm}
\newcommand{\Nm}[1]{\ensuremath({\Nmop}\;{#1})}

\newtrm{\Skop}{Sk}
\newcommand{\Sk}[1]{\ensuremath({\Skop}\;{#1})}

\newtrm{\Ikop}{Ik}
\newcommand{\Ik}[1]{\ensuremath({\Ikop}\;{#1})}

\newtrm{\Akop}{Ak}
\newcommand{\Ak}[1]{\ensuremath({\Akop}\;{#1})}

\newtrm{\Svop}{Sv}
\newcommand{\Sv}[1]{\ensuremath({\Svop}\;{#1})}

\newtrm{\Avop}{Av}
\newcommand{\Av}[1]{\ensuremath({\Avop}\;{#1})}

\newtrm{\Mgop}{Mg}

\newsf{\unv}{unv}

\newcommand{\tv}{\ensuremath{\tau}\xspace}
\newcommand{\rolvalop}{\tv}
\newcommand{\rolval}[2]{(\rolvalop \ {#1} \ {#2})}

\newsf{\pr}{pr}

\newrm{\proc}{proc}
\newrm{\procs}{procs}

\newrm{\Loc}{Loc}
\newcommand{\loc}{v}

\newrm{\stmt}{stmt}
\newrm{\event}{event}
\newrm{\bind}{bind}
\newrm{\chck}{check} 

\newrm{\Expr}{Expr}

\newcommand{\newexp}{\newsf}

\newexp{\Pairop}{Pair}
\newcommand{\Pair}[2]{\ensuremath({\Pairop}\;{#1}\;{#2})}

 \newexp{\Encrop}{Encr}
\newcommand{\Encr}[2]{\ensuremath({\Encrop}\;{#1}\;{#2})}

\newexp{\Hashop}{Hash}
\newcommand{\Hash}[1]{\ensuremath({\Hashop}\;{#1})}

\newexp{\Frstop}{Frst}
\newcommand{\Frst}[1]{\ensuremath({\Frstop}\;{#1})}
\newexp{\Scndop}{Scnd}
\newcommand{\Scnd}[1]{\ensuremath({\Scndop}\;{#1})}

\newexp{\Decrop}{Decr}
\newcommand{\Decr}[2]{ ( \ensuremath {\Decrop}\;{#1}\;{#2} ) }

\newexp{\Quotop}{Quot}
\newcommand{\Quot}[1]{\ensuremath({\Quotop}\;{#1})}

\newexp{\Genrop}{Genr}
\newcommand{\Genr}[2]{\ensuremath({\Genrop}\;{#1}\;{#2})}

\newexp{\Pubofop}{PubOf}
\newcommand{\PubOf}[1]{\ensuremath({\Pubofop}\;{#1})}

\newexp{\Priofop}{PriOf}

\newexp{\Paramop}{Param}
\newcommand{\Param}[1]{\ensuremath({\Paramop}\;{#1})}

\newexp{\Readop}{Read}
\newcommand{\Read}[1]{\ensuremath({\Readop}\;{#1})}

\newexp{\Recv}{Recv}

\newcommand{\newstmt}{\newsf}

\newstmt{\Bindop}{Bind}
\newcommand{\Bind}[3]{\ensuremath \Bindop \; ( {#1} , {#2}) \; {#3}}

\newstmt{\Checkop}{Check}

\newstmt{\CSameop}{CSame}
\newcommand{\CSame}[2]{\ensuremath({\CSameop}\;{#1}\;{#2})}

\newstmt{\CKyprop}{CKypr}
\newcommand{\CKypr}[2]{\ensuremath({\CKyprop}\;{#1}\;{#2})}

\newstmt{\CSrtop}{CSrt}
\newcommand{\CSrt}[2]{\ensuremath({\CSrtop}\;{#1}\;{#2})}

\newstmt{\CQotop}{CQot}
\newcommand{\CQot}[2]{\ensuremath({\CQotop}\;{#1}\;{#2})}

\newstmt{\CHashop}{CHash}
\newcommand{\CHash}[2]{\ensuremath({\CHashop}\;{#1}\;{#2})}

\newcommand{\samenessop}{\ensuremath \approx_{sm} \xspace}
\newcommand{\sameness}[2]{{#1} \approx_{sm} {#2}\xspace}

\newcommand{\sto}{\ensuremath{\sigma}\xspace}

\DeclareMathOperator{\elementary}{elementary}

\DeclareMathOperator{\trace}{trace}
\DeclareMathOperator{\map}{map}
\DeclareMathOperator{\mapR}{map_R}

\newrm{\Some}{Some}

\newcommand{\tlop}{\ensuremath{\beta}\xspace}
\newcommand{\tl}[2]{(\tlop \; {#1}\; {#2})}

\newcommand{\completion}[1]{\widehat {#1}}

\newcommand{\locnew}{\ensuremath{\loc_{\text{new}}}}

\newcommand{\lil}{\lstinline}

\newcommand{\D}{\ensuremath{\mathcal{D}\xspace}}


%% file: body.tex

\section{Introduction}
\label{sec:intro-prose}



Protocol narrations, colloquially, ``Alice \& Bob'' notation, are a
nearly ubiquitous informal means to describe the intended executions
of cryptographic protocols as sequences of cryptographic message
exchanges among the protocol’s participants.  This notation is
succinct and approachable but it lacks a formal semantics and leaves
many implementation details implicit.

We are interested here in bridging the gap between informal
description and implementation.  To that end we present \molly, a
compiler that transforms a narration in the \cpsa\ input language into
an implementation.   

{\molly} generates executable code in an intermediate language to
handle the transmissions and receptions dictated by the protocol, as
well as any checks that must be done before taking one of these
steps.  %
For example if the protocol specifies reception of a message $m$, and
at runtime the network delivers the value $v_m$, the code must
\begin{enumerate}
\item deconstruct $v_m$ to verify that it has the structure
  required in the protocol specification, e.g.~that it is a properly
  formed encryption using the expected key, and its plaintext is a
  tuple of components of the right kinds, and so on;

\item do a check that whenever some part of $m$ is to be
  equal to some part of another message already processed, the
  corresponding runtime values are indeed equal,

\item store $v_m$ and its components so that they can be used in
  the remainder of the execution as required.
\end{enumerate}
In the dual case, when a message must be prepared and transmitted, the
 code must arrange to build up a suitably structured
value from values already known, together with new values that may be
chosen randomly, such as nonces and session keys.

In addition, we define a denotational semantics for protocol roles
given in the \cpsa\ language, defined in terms of ``transcripts'',
runtime traces of messages as bitstrings.  Our intermediate language
will be seen to have an obvious transcript semantics, and our main
theorem is a correctness theorem relating role transcripts and
transcripts for our intermediate language.

{\molly} is written in Coq's Gallina programming language, and our
correctness result is developed and checked within Coq.  %
Our Coq development is available at %
\url{github.com/dandougherty/Molly}

\subsection
{Outline and Contributions}
\label{sec:contibutions}

There were a number of key goals that shaped our work reported here.  

\emph{Code Generation for an Intermediate Language.}  Functionally,
{\molly} is a partial function from \roles to \emph{procs}; the latter are
essentially the Procs of Ramsdell's Roletran
\cite{ramsdell2021cryptographic}. %
A proc is a sequence of intermediate-level instructions.  Procs are
straightforward to translate into concrete executable programming
languages, but also easy to characterize semantically.  Hence, they
are an attractive target language for {\molly} as a verified compiler.
{\molly} is written in Coq's Gallina specification language, so that 
using Coq's extraction mechanism, we derive code, in Ocaml
for instance,  to produce a \proc from a role.\footnote{%
It is straightforward to translate \procs into code for a conventional
language with a library for cryptographic primitives such as C, Rust,
or Java; we have not pursued this at this point in the project.}

\emph{Transcripts for Roles.} %
The crucial first step in proving correctness of this compilation is
defining the notion of \emph{transcript for a role:} a transcript is a
sequence of runtime values which might arise as a runtime execution of
the role %
(Section~\ref{sec:transcripts-role}). %
Transcripts live in a domain---bitstrings---independent of
symbolic messages.  It seems suitable to say that role transcripts
yield a \emph{denotational semantics} for protocol roles: the meaning
of a role is the set of its transcripts.  Put another way, the meaning
of a role is its set of \emph{observable actions}.

The development of this idea is somewhat subtle, chiefly because we
take randomized encryption seriously.  We feel that our development of
transcripts gives useful perspective on protocols even outside the
context of compiling, and we count it as one of our main
contributions.  We also define the notion of \emph{transcript for a
  proc} (\refsec{sec:transcripts-proc}): this definition is utterly
straightforward.

Take note of the distinction between a semantics for individual roles,
treated here, and a semantics for \emph{protocol executions}, which of
course comprise interactions between individual role executions.
The former, our transcripts, are merely ``sections'' of the latter, the
behaviors observable by a single principal.    

We want to stress the point that without a semantics for roles and
\procs defined independently from each other, it doesn’t even make
sense to claim that a compilation is correct, much less prove such a claim

\emph{Treatment of Randomized Encryption.} Manipulating symbolic terms
is useful only if we maintain a realistic view of the relationship
between terms and actual network messages consisting of bitstrings. %
Because many cryptographic operations are randomized, the bitstrings
are not in one-to-one correlation with the symbolic terms---even given
an assignment of bitstrings to the variables occurring in the terms.
In particular calling an encryption or digital signature primitive on
the same message twice, with the same key, does not yield the same
bitstring result. 

This has the semantic consequence that role transcripts are generated
from \emph{relations} from values to bitstrings (the ``valuations'' of
Section~\ref{sec:valuations}).  In turn, our compiler embodies choices
about which bitstring should be used as the representative of a
symbolic term when more than one is available, and our correctness
proof characterizes the constraints on these strategies.

\emph{Axiomatizing the Runtime.} Our correctness theorem is about
transcripts; transcripts are sequences of runtime values, so we
require some analysis of the runtime.  A key feature of our approach
is that we do not \emph{define} a runtime for our proof, rather we
simply isolate surprisingly few mild assumptions about the runtime we
require for the proof. %
This strategy has the obvious benefits of broadening the applicability
of the result and of identifying the principles that make compilation
correct.

\emph{A Machine-Checked Correctness Theorem.}
Our main theorem, the Reflecting Transcripts theorem
(\refthm{thm:reflecting-transcripts}), %
states that if role \rl is compiled to procedure \pr, any transcript
for \pr is a transcript for \rl.  As explained in
Section~\ref{sec:preserving-transcripts}, %
the converse of Reflecting Transcripts theorem cannot be expected to
hold in the presence of randomized encryption.

\emph{Proof-Theoretic View of Code Generation.}  We have organized our
compilation according to the Dolev-Yao model~\cite{DolevYao83}, by
which we mean that we view the principal executing a role as
manipulating term-structured items according to derivation rules.
Whatever may be said about the strengths and weaknesses of the
Dolev-Yao model as a model of the adversary, this is certainly a
reasonable model of the regular, compliant protocol participants; the
protocol designer does not expect them to have to succeed at
cryptanalysis, for instance, to run the protocol successfully.  Thus,
reflecting much previous work, dating back at least to Paulson and
Marrero et al.~in the 1990s~\cite{Paulson97,ClarkeJhaMarrero98}, we
take a proof-theoretic view of the actions of our compiler.  It emits
code by executing steps that we formalize as inference rules, thus
generating derivations in a Gentzen-Prawitz natural deduction
calculus~\cite{Gentzen35,Prawitz65}.  Details can be found in
Section~\ref{sec:derivability}.

\emph{Compilability and  Executability.}
Several authors have given definitions of ``executability'' of a role,
intending to capture the informal notion of a role providing enough
information to be able to run (notably,  having decryption keys
available when required).
Existing definitions tend to be incomparable across formalisms.
The results of Section~\ref{sec:derivability} allow us to characterize 
the input roles on which compilation \emph{fails}.
And this in turn allows us to motivate and define a notion of
\emph{executability} for a role in terms of the well-known notion of Dolev-Yao
derivability.

\subsection{Related Work}
\label{sec:rel-work}

There is growing interest in the development of verified compilers,
for conventional languages (\eg, \cite{Leroy-Compcert-CACM}) as well
as domain-specific languages (\eg, \cite{pohjola2022kalas}).  To our
knowledge {\molly} is the first compiler for cryptographic protocols with
a machine-checked correctness condition.

Many authors have worked to bridge the gap between protocol narrations
in a semi-formal style and protocol descriptions that are more formal.
We organize the discussion below according a crude partition.  One
category is work that translates protocol narrations to
another---formally defined---language, like process calculus or
multiset rewriting. Typically the payoff is that automated
verification tools can then be used to reason about the protocol.
Another category is tools that compile protocol descriptions into a
conventional programming language.  Often the input to these tools is
already in a formal notation such as spi-calculus, and sometimes the
translation is instrumented with tools that support claims about the
security guarantees of the target program.

Work in the first category includes %
translation of protocol narration %
to CSP \cite{lowe1997casper}, %
to generic intermediate languages 
\cite{denker2000capsl} \cite{modersheim2009algebraic},
\cite{almousa2015alice}, %
to Multiset Rewriting 
\cite{jacquemard2000compiling}
\cite{keller2014converting}, %
to symbolic representations of principals' knowledge
\cite{mccarthy2008cryptographic}
\cite{basin2015alice} %
(annotated with unifiability conditions in
\cite{chevalier2010compiling}), %
and to Pi-calculus and variants 
\cite{caleiro2005deconstructing} 
\cite{caleiro2006semantics} 
\cite{briais2007formal}.   %

%

In most of the works above the authors offer their work as supporting
an ``operational semantics'' for roles, and here we can identify an
interesting difference between our work and theirs.
In the works above we can identify two broad operational semantics
approaches. %
In one case the job is to define the \emph{activities that an agent
  takes} to implement the protocol: constructing and deconstructing
messages, certain checks, etc. 
In the other case one defines the possible \emph{executions of the
  protocol:}  sometimes a notion of trace is defined, tracking the
evolution of symbolic representation of principal's knowledge,
(Cremers \cite{cremers2005operational} is a detailed development of
this perspective) or alternatively the semantics is implicit in the
semantics of the target formalism (pi-calculus, multiset rewriting,
etc).

Our work cuts across these two functions. %
The sequence of activities that an agent takes to implement a given
protocol role is precisely the \proc built by our compilation. %
And, our transcripts capture the executions of the protocol not in
terms of symbolic terms, but rather in terms of bitstrings.

We will see that our procs support an obvious notion of transcript.
Then as noted above, since roles and \procs have a common target domain for
their semantics it makes sense to compare the meaning of a role and
the meaning of a \proc.    Prior work offers no formal proof of
correctness of the translation process;
our main contribution is
a machine-checked proof of a theorem doing just that.

An intriguing aspect of the work in Caleiro, Vigano, and Basin
\cite{caleiro2006semantics} is that they employ a notion of
incremental symbolic runs as a basis for a \emph{denotational}
semantics.  Each state of a run reflects the information known by the
principals; the run itself models how information grows.  The rules
for evolution of these runs look very much like our saturation process
in \refsec{sec:saturation} for building the bindings of a \proc!  The
difference is that for them the process of recording the growth of
principals' knowledge \emph{is} a semantics of a role, while for us
this process is the essence of \emph{compiling}, not execution.  As
explained earlier, our denotational semantics is grounded in the world
of bitstrings.

Arquint et al \cite{arquint2022sound}, \cite{Arquint:2023aa} are
mainly interested in verification, and present a tool that is not
really a protocol compiler: it starts with a Tamarin model and
generates a set of \emph{I/O specifications} in separation logic.  But
their main correctness result has an interesting relationship to ours.
An I/O specification is a set of permissions needed to execute an I/O
operation \cite{penninckx2015sound}. %
Then (quoting \cite{arquint2022sound}) ``traces can intuitively be
seen as the sequences of I/O permissions consumed by possible
executions of the programs that satisfy it.''  They prove that if
abstract Tamarin model $M$ is translated to a set $S$ of I/O
specifications, then any concrete implementation satisfying $S$
refines $M$ in terms of trace inclusion.  This is closely analogous
to our Reflecting Transcripts theorem
\ref{thm:reflecting-transcripts}, with logical properties standing in
for bitstrings.

We now turn to the category of projects that are primarily
focused on generating implementations from protocol specifications.
Although the work in \cite{almousa2015alice} is not principally
focused in this way, that paper reports a translation from its
intermediate language SPS into JavaScript.

Tobler and Hutchinson \cite{tobler2005generating} built the Spi2Java
tool, which builds a Java code implementation of a protocol specified
in a variation of the Spi calculus.  There is no proof that the
semantics of the input specification is preserved by the translation.

Backes, Busenius, and Hritcu %
\cite{backes2012development} %
developed Expi2Java, which translates models written in the Spi
calculus \cite{abadi1997calculus} into Java.  They formalized their
translation algorithm in Coq and proved that the generated programs
are well-typed if the original models are well-typed.

Modesti %
\cite{modesti2014efficient,modesti2016anbx} developed the ``AnBx''
compiler, which generates Java code from protocols written in an
\AnB-style notation.  The tool
generates certain consistency checks and annotates the translation
with applied pi-calculus expressions to permit a ProVerif
\cite{DBLP:journals/ftsec/Blanchet16}
verification
that security goals are met by the Java code.

The JavaSPI tool of Sisto, Copet, Avalle and Bronte %
\cite{sisto2018formally} %
starts with code in a fragment of the
language that corresponds to applied pi-calculus
\cite{abadi2017applied}. %
The tool can symbolically execute this code in the Java debugger,
formally verify it using ProVerif, 
eventually refine to an Java implementation of the protocol.  They
prove that a simulation relation relates the Java refined
implementation to the symbolic model verified by ProVerif.

Spi2Java, Expi2Java and JavaSPI require the user to provide input in a more
demanding formalism than the familiar \AnB-style.
A benefit of all of the systems in this category compared with {\molly} is
the fact that they produce code for the ubiquitous Java platform.
On the other hand, for none of these systems is there a proof of correctness of the compilations
themselves.

Ramsdell's Roletran compiler
  \cite{ramsdell2021cryptographic}
has functionality and overall goals quite close to that of {\molly}.   The input 
is a \cpsa\ specification of a role, and the output is a program in an
intermediate language designed to be readily translated to a
conventional language:   our input and output languages are inessential
variations on Roletran's.   %
The main
correctness claim of Roletran is that ``the procedure produced by
Roletran is faithful to [the] strand space semantics [of the input
role].''
Roletran does not have a machine-checked  proof of its global correctness
claim, but the distribution does the following  interesting thing.
Coq scripts are provided that can check, for a given role \rl, that
the procedure generated by the tool is correct (according to the
symbolic-trace semantics).

Differences between {\molly} and Roletran include the facts that Roletran's
strand space semantics is in terms of \emph{symbolic traces} for a
role, as opposed to our role semantics based on
bitstrings, and that we provide a machine-checked proof of a uniform
correctness theorem.  Roletran has some restrictions on the messages
that can appear in roles compared with {\molly}.

Roletran is also the inspiration for a fully usable framework called
\emph{Zappa} that augments a role compiler with many ingredients
needed for a runtime system, including runtime message formats based
on ASN.1 encodings.  The Zappa compiler generates procedures in the
Rust programming language for a substantial extension of the source
syntax considered here and in Roletran.  Correctness proofs have not
been considered for the extensions.
Cryptographic libraries
available in Rust may be linked in.   

\subsection{Road map}

We introduce the main ideas of compilation, and the notation for roles
and \procs, through a series of examples in Section~\ref{sec:examples}.
Section~\ref{sec:preliminaries} gives preliminary definitions, and %
Section~\ref{sec:runtime} lists our assumptions about the runtime. %
Section~\ref{sec:compilation} explains the compiler at a high level; %
section~\ref{sec:saturation} presents the details of the algorithm. %
Section~\ref{sec:structure-procs} presents some results about the \procs
constructed by the compiler.
Section~\ref{sec:transcripts} defines our formal models of execution of a role and of a
\proc, and %
section~\ref{sec:relating} proves our main theorem, the Reflecting Transcripts theorem.

\section{Overview and Examples}
\label{sec:examples}

The input to {\molly} is a description of one role of a cryptographic
protocol.  %
There are a variety of notations to specify protocols; here we use
the input language of the protocol analysis tool \cpsa. %
The output of {\molly} is a program in an intermediate language which is
readily translatable to a program in a language such as C or Rust. %
A program in this intermediate language will be called a \emph{proc}.
Roles and \procs are defined precisely in Sections~\ref{sec:terms} and
\ref{sec:procs}.

Our execution model is that \procs read and write runtime values on
channels, maintain a local store of runtime values, and execute
commands that update their state.

To provide a little more detail into the way a role specification is
translated to  \proc code, and to introduce some of the
subtleties that arise in this translation, we discuss a series
of small examples.


\begin{example}   \label{eg:init1}
  Here is a simple role.  %
  It takes three parameters, a channel and two atoms of sort Data. %
  It sends over the channel an ordered  pair constructed from the data,
  and subsequently expects to receive one of the data items over the
  same channel.
  \begin{align*}
    & \Prm {\Ch 1}; \\
    & \Prm {\Dt 1}; \\
    & \Prm {\Dt 2}; \\
    & \Snd {\Ch 1} {\Prr {\Dt 1} {\Dt 2}} ; \\
    & \Rcv {\Ch 1} {\Dt 2}  
  \end{align*}    
  Below is a sequence of instructions in our \proc language to execute this role.
  (Line numbers are added to the display of the \proc here for
  facilitate commentary.)

  For each parameter, a location is allocated, the appropriate
  value is stored; and a runtime check is emitted to ensure that the
  parameter value has the expected sort 
  (for example, lines \ref{line:iptstart}--\ref{line:iptend} for the first parameter).

  For the transmission:  we  create a location to store this pair-value
  (line~\ref{line:bindpr}) and send the value in this location over channel 1
  (line~\ref{line:snd}).

  For the reception: the reception of data from channel 1 and
  implicit declaration of the location for the data are recorded
  (line~\ref{line:recv}),  the data stored in that location 
  (line~\ref{line:bind}), and a check is done
  that the incoming value is the same as the expected one (line~\ref{line:same}).

  \begin{lstlisting}[frame=single,backgroundcolor=\color{yellow!10},
    caption={for \refeg{eg:init1}}]
    (* first parameter *)
    Evnt (Prm (L 1));            (*@ \label{line:iptstart} @*)
    Bind (Ch 1, L 1) (Param 1);
    Csrt (L 1) Chan;              (*@ \label{line:iptend} @*)
    (* second parameter *)
    Evnt (Prm (L 2));
    Bind (Dt 1, L 2) (Param 2);
    Csrt (L 2) Data;
    (* third parameter *)
    Evnt (Prm (L 3));
    Bind (Dt 2, L 3) (Param 3);
    Csrt (L 3) Data;
    (* the send action *)
    Bind (Pr (Dt 1) (Dt 2), L 4) (Pair (L 2) (L 3));  (*@ \label{line:bindpr} @*)
    Evnt (Snd (L 1) (L 4));      (*@ \label{line:snd} @*)
    (* the receive action *)
    Evnt (Rcv (L 1) (L 5));      (*@ \label{line:recv} @*)
    Bind (Dt 2, L 5) (Read 1);   (*@ \label{line:bind} @*)
    Same (L 5) (L 3)             (*@ \label{line:same} @*)
  \end{lstlisting}
\qed \end{example}

\begin{example}
  \label{eg:resp1}

  If we view role init1 as initiating a session, the following is a role
  that could serve as responder.
  
  It takes one parameter, the channel.  The two parameters from role
  init1 are now (expected to be) components of the single value
  received by resp1.  Assuming the reception is accepted, the second
  component is sent back on channel 2.
  \begin{align*}
    & \Prm {\Ch 1}; \\
    & \Rcv {\Ch 1} {\Prr {\Dt 1} {\Dt 2}}; \\
    & \Snd {\Ch 1} {\Dt 2} 
  \end{align*}
  Below is the \proc generated by the compiler.
 The parameter $\Ch{1}$  is processed just as before.
  The reception of $(\Prr {\Dt 1} {\Dt 2}$ begins with the binding of
  the receive runtime value to location \lil{L 2}.  %
  We next generate code that will check whether we have received a
  suitable value: indeed the actual value received might not be
  a pair at all, or it might be a pair of values not of sort Data.
  
  We deconstruct the received value by binding locations to the result of
  operators \Frstop and \Scndop (lines~\ref{line:frst} and
  \ref{line:scnd}). %
  Crucially, these operators can fail at runtime:  the expression \Frstop
  fails if given a non-pair as input, and otherwise returns the first
  component of the pair; similarly for \Scndop.

  Once those deconstructions are done we expect to be left with atomic
  values; we emit the appropriate sort-check statements; these will of course
   fail on input of the wrong sort.

  For the \Scndop transaction: since the value to be sent will already
  available in location 4, we simply emit the $\Sndop$ event statement.

  \begin{lstlisting}[frame=single,backgroundcolor=\color{yellow!10},
    caption=for \refeg{eg:resp1}
    ]
    Evnt (Prm (L 1));
    Bind (Ch 2, L 1) (Param 1);
    Csrt (L 1) Chan;

    Evnt (Rcv (L 1) (L 2));
    Bind (Pr (Dt 1) (Dt 2), L 2) (Read 1);
    Bind (Dt 1, L 3) (Frst (L 2)); (*@ \label{line:frst} @*)
    Bind (Dt 2, L 4) (Scnd (L 2)); (*@ \label{line:scnd} @*)
    Csrt (L 4) Data;    
    Csrt (L 3) Data;

    Evnt (Snd (L 1) (L 4))
  \end{lstlisting}

\qed \end{example}

\begin{example}  \label{eg:gen-hash}
  \hfill
  \begin{align*}
    & \Prm {\Ch 1}; \\
    & \Snd {\Ch 1} {\Dt 1}; \\
    & \Rcv {\Ch 1} {\Hs {\Dt 1}}
  \end{align*}  
  There are  two new things to talk about here.

  First, we want to send a value corresponding to $(\Dt 1)$, which is
  a value we do not have stored in advance. %
  Our \procs have an operator $\Genrop$ that produces new data of a specified
  sort (line~\ref{line:gen}). %
  In our \proc language this operator also takes a natural number
  parameter, to achieve referential transparency:  different occurrences
  of a $\Genrop$ expression always denote the same value.   When \procs
  are translated to stateful languages like C or Rust, this parameter
   might disappear.

  As a programming design decision, {\molly} emits all the necessary
  $\Genrop$ statements at the start of the code generation process.
  Thus the code for  generation of $\Dt 1$ appears even before the
  processing of parameters.

  Second, we need to ensure that the value received for the %
  ${(\Hs {\Dt 1})}$ term really is the hash of the value received for
  the ${(\Dt 1)}$.  But $\Hsop$ has  no analogue to $\Frstop$ or
  ${\Scndop}$ in our \proc language of course: it is computationally infeasible to
  destruct the hash operator.  Instead, we rely on the fact that we do
  have the value for %
  ${(\Dt 1)}$ stored, the one we generated for the transmission. %
  We compute the hash of \emph{that} value (line~\ref{line:hashI}) %
  and compare that to the value
  received (line \ref{line:same-check}).

  \begin{lstlisting}[frame=single,backgroundcolor=\color{yellow!10},
    caption={for \refeg{eg:gen-hash}}]
    (* allocation for fresh value *)
    Bind (Dt 1, L 1)  (Genr 1 Data); (*@ \label{line:gen} @*)
    
    (* first parameter *)
    Evnt (Prm (L 3));
    Bind (Ch 1, L 3) (Param 1);
    Csrt (L 3) Chan;

    (* transmission *)
    Evnt (Snd (L 3) (L 1));

    (* reception *)
    Evnt (Rcv (L 3) (L 4));
    Bind (Hs (Dt 1), L 4) (Read 1);   
    Bind (Hs (Dt 1), L 2) (Hash (L 1)); (*@ \label{line:hashI} @*)
    Csame (L 4) (L 2) (*@ \label{line:same-check} @*)
  \end{lstlisting}
\qed \end{example}

\begin{example} \label{badhash}
  \label{eg:badhash}

  Here is a role that our compiler will decline to compile.
  \begin{align*}
    & \Prm {\Ch 1};  \\
    & \Rcv {\Ch 1} {\Hs {\Dt 1}}; \\
    & \Snd {\Ch 1} {\Dt 1}    
  \end{align*}
  Since our \proc does not have any stored value corresponding to %
  ${\Dt 1}$ at the time of the reception, we cannot emit any code
  to verify at runtime that the received value is a correct hash.
  It is presumably possible to check that the value received is a hash
  of \emph{something,} so it might be tempting to call that something %
  ``${\Dt 1}$'' but we
  certainly could not use the term 
  ${\Dt 1}$ elsewhere in the role since we have no information about
  its value.
\qed \end{example}

In Example \ref{badhash} one might be tempted to
  expect that {\molly} would use $\Genrop$ to generate a location bound to
  ${\Dt 1}$ initially, so that we could use the same strategy as
  there (computing the hash of the known value and comparing equality).
  But the operator ${\Genrop}$ is only performed on locations
  corresponding to terms that \emph{originate} in the role, that is,
  terms whose first occurrence in the role is part of a transmission.
  The notion of origination is a central concept in the analysis of
  protocols by \cpsa.  

\paragraph{Public-Key Encryption}

\begin{example}
  \label{eg:encr1}
  \hfill
  \begin{align*}
    & \Prm {\Ch 1}; \Prm {\Ik {\Av 2}}; \\
    & \Rcv {\Ch 1} {\En {\Nm 0} {\Ak {\Av 2}}}; \\
    & \Snd {\Ch 1} {\Nm 0}    
  \end{align*}
  The term ${\Ak {\Av 2}}$ %
  denotes an asymmetric key generated from the variable $\Av 2$.  %
  By convention this denotes the public key associated with this
  variable.  The term $\Ik {\Av 2}$ denotes the corresponding private
  key; these two terms make a \emph{key pair}.

  The fact that $\Ik {\Av 2}$ %
  is a parameter reflects the idea that this key is known to the principal
  executing this role; the role can decrypt  an incoming messages encrypted with the
  key-partner of this term.    This is precisely what happens in line
  \ref{line:decr}:
  the encryption is stored in location $(L 3)$ and  the decryption key 
  is expected to be stored in location $(L 2)$, we allocate a new
  location $(L 4)$ and store the result of the decryption there
  (line~\ref{line:decr}).

  It is important to note that just as with \Frstop %
  and \Scndop, a \Decrop operation can fail.
  That is, if the runtime values appearing at a reception do not fit
  the specification embodied in the role,  execution will halt.





 \begin{lstlisting}[frame=single,backgroundcolor=\color{yellow!10},
    caption={for \refeg{eg:encr1}}]
    Evnt (Prm (L 1));
    Bind (Ch 1, L 1) (Param 1);
    Csrt (L 1) Chan; 

    Evnt (Prm (L 2));
    Bind (Ik (Av 2), L 2) (Param 2);
    Csrt (L 2) Ikey; 

    Evnt (Rcv (L 1) (L 3));
    Bind (En (Nm 0) (Ak (Av 2)), L 3) (Read 1);
    Bind (Nm 0, L 4) (Decr (L 3) (L 2));  (*@ \label{line:decr} @*)
    Csrt (L 4) Name; 

    Evnt (Snd (L 1) (L 4))
\end{lstlisting}

\qed \end{example}

\begin{example}[Encryption 2]
  \label{eg:encr-fail}   \hfill

  The following role will be rejected by the compiler. 
  \begin{align*}
    & \Prm {\Ch 1}; \\
    & \Rcv {\Ch 1} {\En {\Nm 0} {\Ak {\Av 2}}}; \\
    & \Snd {\Ch 1} {\Nm 0}
  \end{align*}
  There is no way to decrypt the received message, since the key
  partner of $\Ak {\Av 2}$ is not a parameter nor can it be
  constructed from the data available at the time of the reception.
  \hfill
\qed \end{example}

\paragraph{Treating Randomized Encryption}

\begin{example}
  \label{eg:encr-sym}   
  As a simple first example for symmetric encryption consider a role
  that expects an encryption whose key is the hash of a known value and
  replies with the decrypted plaintext.
  \begin{align*}
    & \Prm {\Ch 1};  \Prm{\Dt 2}; \\
    & \Rcv {\Ch 1} {\En {\Nm 0} {\Hs {\Dt 2}}}; \\
    & \Snd {\Ch 1} {\Nm 0}    
  \end{align*}
  A suitable \proc follows.  For readability we introduce some
  whitespace in the listing to make it easier to focus on the
  treatment of encryption.

  Here %
  \begin{itemize}
  \item \lil{(L 1)} stores the value of the channel, and 
    \lil{(L 2)} stores the value of $\Dt 2$
  \item In line~\ref{line:make-hash} we construct the hash of $\Dt 2$
    in preparation for constructing and storing it in \lil{(L 3)}
  \item the received value is stored in \lil{(L 4)}.
  \item In line~\ref{line:decr2} we decrypt the received value with the
    key we constructed in line~\ref{line:make-hash}.  If this decryption
    succeeds (because the reception was indeed encrypted with the value
    of $\Hs {\Dt 2}$ we constructed) we store the result in %
    \lil{(L 5)} and subsequently send it.
  \end{itemize}
  \begin{lstlisting}[frame=single,backgroundcolor=\color{yellow!10},
    caption={for \refeg{eg:encr-sym}}]
    Comm "input (Ch 1)";
    Evnt (Prm (L 1));
    Bind (Ch 1, L 1) (Param 1); 
    Csrt (L 1) Chan;
    Comm "input (Dt 2)"; 
    Evnt (Prm (L 2));
    Bind (Dt 2, L 2) (Param 2);

    Bind (Hs (Dt 2), L 3) (Hash (L 2)); (*@ \label{line:make-hash} @*)

    Csrt (L 2) Data;
    Comm "receiving (En (Nm 0) (Hs (Dt 2)) ) on (Ch 1)";
    Evnt (Rcv (L 1) (L 4));

    Bind (En (Nm 0) (Hs (Dt 2)), L 4) (Read 1);  (*@ \label{line:rcv-encr} @*)
    Bind (Nm 0, L 5) (Decr (L 4) (L 3)); (*@ \label{line:decr2} @*)

    Csrt (L 5) Name; 
    Comm "sending (Nm 0) on (Ch 1)";
    Evnt (Snd (L 1) (L 5))
  \end{lstlisting}
\qed \end{example}

In the symbolic model messages are represented as elements of a free
algebra for an equational theory, the \emph{message algebra}.  
In particular encryption is named by a term %
$\En{m}{k}$.  %

But we assume that our runtime implements \emph{randomized} symmetric
encryption.\footnote{ We do not attempt to \emph{model} randomized
  encryption, that is, we do not formalize any mechanisms in the
  runtime or in the term algebra that reflect computational processes
  of encryption and decryption that use randomization.}

Thus we have a mismatch between the traditional algebraic semantics of
term algebras and the runtime semantics.

For example, if a term %
$\En{m}{k}$ occurs more than once in a protocol description it may be
considered to denote different runtime values at different
occurrences.

So the problem arises of how to define a semantics for
the message algebra that is faithful to runtime semantics.
(Say that we aren't going to try to model the probabilistic aspect,
just the possibilistic aspects.)

By the way, these considerations arise prior to the idea of writing a
compiler: they are just about the establishing the meanings of
cryptographic terms.

Here is a series of examples that highlight some ways in which
randomized encryption poses interesting design choices for our compiler.

\begin{example}[A simple failure]
  \label{eg:fail-1}
  The following role will be rejected by the compiler, 
  for the same reason as in
  \refeg{eg:encr-fail}
  \begin{align*}
    & \Prm {\Ch 1}; \\
    & \Rcv {\Ch 1} {\En {\Dt 3} {\En {\Dt 1} {\Dt 2} }}    
  \end{align*}
  There is have no way to decrypt the received message, since we
  cannot construct %
  $\En {\Dt 1} {\Dt 2}$ 
  as a decryption key from the data available.
\qed \end{example}

Compare the next example with \refeg{eg:encr-sym}.

\begin{example}\hfill
  \label{eg:fail-2}
  \begin{align*}
    & \Prm {\Ch 1}; \Prm {\Dt 1}; \Prm {\Dt 2}; \\
    & \Rcv {\Ch 1} {\En {\Dt 3} {\En {\Dt 1} {\Dt 2} }}     
  \end{align*}
  This role can be compiled, since  we can  construct %
  $\En {\Dt 1} {\Dt 2}$.
  But the resulting code has negligible probability
  of executing successfully: the runtime value of
  $\En {\Dt 1} {\Dt 2}$ 
  we construct from the local data 
  (to be used as decryption key) is unlikely to be the runtime value of
  $\En {\Dt 1} {\Dt 2}$ 
  that was 
  used by a peer as the encryption key to form
  $\En {\Dt 3} {\En {\Dt 1} {\Dt 2}}$.

\qed \end{example}

Another example, similar to the previous but with yet another subtlety.

\begin{example} \hfill
  \label{eg:fail-3}
  \begin{align*}
    & \Prm {\Ch 1}; \Prm {\Dt 1}; \Prm {\Dt 2}; \\
    & \Snd {\Ch 1} {\En {\Dt 1} {\Dt 2} } ; \\
    & \Rcv {\Ch 1} {\En {\Dt 3} {\En {\Dt 1} {\Dt 2} }}    
  \end{align*}
  This role can be compiled, since  we can (and did!) construct %
    $\En {\Dt 1} {\Dt 2}$ 

  Under  a strictly mathematical  analysis the likelihood of successful
  execution here is the same as in the previous example.
  But it is plausible to imagine that the received message would
  instantiated as the same
  runtime value for
  $\En {\Dt 1} {\Dt 2}$ 
  as was sent in the second message,
  for example if this were the  conventional understanding of the
  protocol.   But the protocol specification language provides no way to
  express this convention.
\qed \end{example}

Section~\ref{sec:executability} explains why this phenomenon is, in a
sense, inevitable.


\section{Preliminaries}
\label{sec:preliminaries}

\subsection{Actions}
\label{sec:actions}       

 The \Actop parameterized data type is a useful device for tying
 together several recurring constructions.
 The constructor \Prmop builds parameters, %
 \Retop builds return values, %
 and \Rcvop and \Sndop builds values received and sent.
 
 We will shortly define  roles, procs, and runtime with respective
 carriers Terms, Locations and runtime values.
 Expressions of type $\Act{\Term}$ will denote symbolic terms as
 parameters, sent messages, etc; %
 expressions of type $\Act{\Loc}$ will denote locations where 
 parameters, sent messages, etc are stored;
 expressions of type $\Act{\rtval}$ will denote runtime values used as 
 parameters, sent messages, etc.

 Having a polymorphic datatype to refer to corresponding construction
 allows a uniform treatment of important correspondences.

\begin{definition}
  If $X$ is a Type, $Act\ X$ is the type whose constructors are
  \begin{align*}
    \Prmop &: X \to Act \ X  \\
    \Retop &: X \to Act \ X  \\
    \Rcvop &: X \to X \to Act \ X \\ 
    \Sndop &: X \to X \to Act \ X \\
  \end{align*}
\end{definition}
 
\subsubsection{Mapping}
\label{sec:mapping}

Map a function or a relation over \Actop.

\begin{definition}[Mapping over \Actop] 
  Let $r: X \to \pow{Y}$ be a relation from $X$ to $Y$. %
  The relation
  $\actmap{r}$ from $ \Act{X}$ to  $\Act{Y}$ is the natural extension of $r$
  to $\Act{X}$:
  \begin{align*}
    \actmap{r}( \Prm{x}\; \Prm{y} ) &\text{ holds if } (r \; x \; y) 
    \\
    \actmap{r}( \Ret{x}\; \Ret{y} ) &\text{ holds if } (r \; x \; y) 
    \\
    \actmap{r}( \Rcv{x_1}{x_2}\; \Rcv{y_1}{y_2} ) 
                                    &\text{ holds if } 
                                      (r \; x_1 \; y_1) \text{ and } (r \; x_2 \; y_2)
    \\
    \actmap{r}( \Snd{x_1}{x_2}\; \Snd{y_1}{y_2} ) 
                                    &\text{ holds if } 
                                      (r \; x_1 \; y_1) \text{ and } (r \; x_2 \; y_2)
  \end{align*}
\end{definition}

As a special case (modulo notation) we observe that %
if $f: X \to Y$ then we have the function
$\actmap{f} : \Act{X} \to \Act{Y}$ with
\begin{align*}
  \actmap{f}(\Prm{x}) &= \Prm{(f x)} \\
  \actmap{f}(\Ret{x}) &= \Ret{(f x)} \\
  \actmap{f}(\Rcv{x_1}{x_2}) &= \Rcv{(f x_1)}{ (f x_2)} \\
  \actmap{f}(\Snd{x_1}{x_2}) &= \Snd{(f x_1)}{ (f x_2)}
\end{align*}

\begin{definition}[Mapping over a list]
  \label{def:list-map}
  Let $r$ be a relation from $X$ to $Y$. %
  The relation
  $\mapR {r}$ is the relation from lists of $X$ to lists of $Y$
  defined by:
  \begin{align*}
    \map {r} \; [x_1, \dots x_n] \; [y_1, \dots y_n] 
    &\text{ if } 
      \forall \; 1 \leq i \leq n, \; (r\ x_i\ y_i)
  \end{align*}

  Of course when $r$ is a function $\mapR$ is the standard ``$\map$'' operation
  mapping a
  function over a list.
\end{definition}

\begin{lemma} \label{lem:act-map-monotone}
  If $r$ and $r'$ are relations from $X$ to $Y$
  with $r \subseteq r'$ then for all lists $xs$ and $ys$, 
  $( \mapR {r} \; xs ys)$
  implies   $(\mapR {r'} \; xs \ ys).$
\end{lemma}
\begin{proof}
  easy.
\end{proof}

\begin{notation}
  If $r_{XY}$ is a relation from $X$ to $Y$ and 
  $r_{YZ}$ is a relation from $Y$ to $Z$ then
  $(r_{XY} ; r_{YZ})$ denotes  the relational composition of 
  $r_{XY}$ with $r_{YZ}$, a relation from $X$ to $Z$.
\end{notation}

\begin{lemma} \label{lem:act-map-composition}
  If $r_{XY}$ is a relation from $X$ to $Y$ and 
  $r_{YZ}$ is a relation from $Y$ to $Z$ then
  \[
    \actmap {(r_{XY} ; r_{YZ}) } = 
    \actmap{r_{XY}}  \ ; \    \actmap{r_{YZ}}
  \]
\end{lemma}
\begin{proof}
  easy.
\end{proof}

\subsection{Sorts}
\label{sec:sorts}
Symbolic terms, \proc expressions, and runtime values obey a common
sort discipline.
\begin{definition}
The  \emph{sorts} are
\begin{align*}
 & \Chan \qquad
 \Data \qquad
 \Name\qquad
 \Text \qquad
 \Skey \qquad
 \Akey \qquad
 \Ikey \\
 & \Mesg 
\end{align*}
The \emph{base} sorts are the sorts other than \Mesg.
\end{definition}

\subsection{Terms and Roles}
\label{sec:terms}

We begin with a set of atoms, each of which has a sort. %
We close under the operations of %
pairing, encryption, hashing, and quotation.

It is convenient to define raw symmetric and asymmetric keys as auxilliary
notions, which will be wrapped to form terms.
  \begin{align*}
   \skeys &\eqdef \{ \Sv{n}  \mid n \in \nat \}
    \\
    \akeys &\eqdef \{ \Av{n} \mid n \in \nat \}
  \end{align*}
  
  The set of terms is defined inductively by the following
  constructors.
  \begin{align*}
    \Chop &: \nat \to \Term &\text{channels} \\
    \Txop &: \nat \to \Term &\text{text } \\
    \Dtop &: \nat \to \Term &\text{data } \\
    \Nmop &: \nat \to \Term &\text{names } \\
    \Skop &: \skeys \to \Term &\text{symmetric keys } \\
    \Akop &: \akeys \to \Term &\text{asymmetric keys } \\
    \Ikop &: \akeys \to \Term &\text{asymmetric keys } \\
    \Qtop &: \strings \to \Term &\text{quotation } \\
    \Prrop &:  \Term \to \Term \to \Term &\text{pairs } \\
    \Enop &:  \Term \to \Term \to \Term &\text{encryptions } \\
    \Hsop &: \Term \to \Term &\text{hashes } \\
    \Mgop &: \nat \to \Term &\text{generic variable} 
  \end{align*}

The \emph{elementary} terms  are the terms \emph{other than} those
whose top-level constructor is one of %
\Prrop, \Enop, \Hsop, or \Qtop.   

A \emph{symbolic key pair} is an ordered pair consisting of a private key and
a public key (in that order) which are inverses for asymmetric
encryption.    Syntactically a symbolic key pair takes the form 
\[ ( \Ik {\Av {n}} , \Ak {\Av {n}} )
\]
Either of the two parts uniquely determines the other,
as evidenced by the symbolic syntax, but  when terms are interpreted
by runtime values, the public can be feasibly computed from the
private part, but not vice-versa.

\subsubsection{Sorts for Terms}
\label{sec:term-sort}

\begin{definition}  
 Each term is assigned a sort.
\begin{align*}
  \sortof \Ch{x} &= \Chan &
  \sortof \Tx{x} = \Text \\
  \sortof \Dt{x} &= \Data &
  \sortof \Nm{x} = \Name \\
  \sortof \Sk{x} &= \Skey &
  \sortof \Ak{x} = \Akey \\
  \sortof \Ik{x} &= \Ikey &
  \text{otherwise: } \sortof t = \Mesg 
\end{align*}
\end{definition}

\subsubsection{Term Inverse}
\label{sec:term-inverse}
The inverse function $\inv{t}$ on algebraic terms $t$ is not named by
a constructor, but it is definable.
\begin{definition}
\label{def:rt-inverse}
The \emph{inverse} of a term $t$, denoted $\inv{t}$ is defined as follows.
\begin{itemize}
\item If $(t_1, t_2)$  is a symbolic key pair then 
  $\inv{t_1} = t_2 $ and    $\inv{t_2} = t_1 $ 
\item $\inv{t} = t$ for all $t$ not part of a key pair
\end{itemize}  
\end{definition}

\subsubsection{Roles}
\label{sec:roles}
A role of a protocol will specify the parameters, the messages
to be sent and received, and the outputs.   Our \roles are essentially
the roles of \cpsa, though in \cpsa\ a role
description also defines which terms are to be assumed to be ``uniquely
originating.''   But this information---crucial to
\emph{analysis}---is not relevant to compiling.   If we omit this
aspect of \cpsa\ role specifications  the
constructors of the type \Actop capture everything we need about a
role, leading to the following simple definition.
\begin{definition}[Role]
  A \emph{\role} is a list of $\Act{\Term}$.
\end{definition}

\subsection{Procs}
\label{sec:procs}

Procs are an intermediate language for representing straight-line
cryptographic programs.    It will be clear that a \proc can be readily
translated---with the help of a suitable cryptographic library---into
a conventional imperative program.


We start with  a set \textbf{Loc} of $\emph{locations}$ for storing runtime
values.

\paragraph*{Expressions}
An  \emph{Expression} is built from locations using certain
  operators that mirror the runtime operators presented  in
  \refsec{sec:runtime}.
  \begin{align*}
    \Pair{\loc_1}{\loc_2}  &: \; \text{pairing values} \\
    \Frst{\loc}     & : \; \text{first projection out of a pair} \\
    \Scnd{\loc}      & : \; \text{second projection out of a pair} \\
    \Encr{\loc_1}{\loc_2}  & : \; \text{encryption} \\
    \Decr{\loc_1}{\loc_2}   & : \; \text{decryption} \\
    \Hash{\loc}    & : \; \text{hashing a value } \\
    \Quot{s}   & : \; \text{a string constant} \\
    \PubOf{\loc}   & : \; \text{public-key partner of a private key} \\
    \Genr{n}{srt} & : \; \text{the $n$th value generated (to be of sort
                    $srt$)} \\
    \Param{n}   & : \; \text{the $n$th input parameter}\\
    \Read{n}   & : \; \text{the $n$th value read from any channel}
  \end{align*}

  \paragraph{Statements}
  A \emph{\proc} is a sequence of \emph{statements.}
  Each statement is an \emph{Event}, a \emph{Bind}, or a \emph{Check}.
  \begin{enumerate}
  \item An {Event} is one of the following four forms
    \begin{enumerate}
    \item    
    $ \Prm{\loc}$ : an input parameter is stored in location \loc

  \item   
      $\Ret{\loc}$ : the value in location \loc is output

    \item $\Rcv{\loc_1}{\loc_2}$ : %
      when the value in $\loc_1$ is a channel, the value received on
      $\loc_1$ is stored in location $\loc_2$
    \item     
      $\Snd{\loc_1}{\loc_2}$ :  
      when the value in $\loc_1$ is a channel, the value stored
      $\loc_2$ is sent on location $\loc_1$
    \end{enumerate}

  \item A {Bind} is of the form
    \[\Bind{t}{\loc}{e}\]
    where
    $t$ is a symbolic term, %
    $\loc$ is a location, and
    $e$ is an {expression}

    This is an assignment statement, storing the value named by $e$
    into the location \loc;    the symbolic term $t$ serves as a type for the
    location $v$.
    
  \item A {Check} is one of the following statement forms.   A check
    is an assertion: if it succeeds, computation simply continues, and
    if it fails, computation halts.
    \begin{enumerate}
     
    \item  $\CSrt{\loc}{s}$ :
       checks that {the value in $\loc_1$  has the sort $s$ } 

     \item  $\CSame{\loc_1}{\loc_2}$ :
        checks that {the values in the given locations are the same} 

     \item  $\CKypr {\loc_1}{\loc_2}$ :
       checks that {the values in the two locs make a private/public
        runtime key pair}

    \item  $\CHash{\loc_1}{\loc_{2}}$ :
      checks that {the value in $\loc_2$ is the hash of the value in
        $\loc_1$} 

    \item   $\CQot{\loc}{s}$ :
      checks that {the value in $\loc$ is the string $s$ }
    \end{enumerate}
  \end{enumerate}




\section{Axiomatizing the Runtime}
\label{sec:runtime}

\molly does not read or generate runtime expressions, but 
the semantics of roles and procs is built on runtime values.
Here we record
our assumptions about  the runtime as an axiomatic theory \rtThy. %
The theorems we prove about \molly will hold
about any implementation of  the runtime
below which qualifies as a model of \rtThy.

\subsection{The Runtime Operators}\label{sec:runtime-operators}

\begin{notation} \hfill
  \begin{itemize}

\item We indicate that a function is partial by writing
its return type as a lifted type (for example the return type of the
operator ${\rtfrst}$ below).

\item When writing about equalities between expressions involving partial
functions in this document we use the following convention:
\[ e_1 \opteq e_2
\]
asserts that both $e_1$ and $e_2$ denote and their values are equal.
\end{itemize}
\end{notation}

\begin{definition}
  The signature for theory \rtThy\ is
\begin{align*}
  \rtpair &: \rtval \to \rtval \to \rtval  &\text{pairing}\\
  \rtfrst &: \rtval \to \option{\rtval} &\text{first projection}\\
  \rtscnd &: \rtval \to \option{\rtval} &\text{second projection}\\
  \rtencr &: \rtval \to \rtval \to \rtval \to \bool
                                           &\text{encryption}\\
  \rtdecr &: \rtval \to \rtval \to \option{\rtval} &\text{decryption}\\
  \rthash &: \rtval \to \rtval  &\text{hashing}\\
  \rtquote &: \strings \to \rtval &\text{quotation} \\
  \rtpubof &: \rtval \to \option{\rtval}  &\text{public key partner
                                            for private key}\\
  \rtgen &: \nat \to \Sort \to \rtval  &\text{value generation}\\
  \rtsrt &: \rtval \to \Sort &\text{check sort}\\
\end{align*}
\end{definition}
We have not included here any operators for processing parameters to a
role or reading values from a channel since we don't analyze these processes.

\subsubsection{Key Pairs and Runtime Inverse}
\label{sec:key-pairs-runtime}

An ordered pair $(r_1, r_2)$ is a \emph{key pair} if it comprises the
private and public parts of an asymmetric key, that is, if 
$\rtpubof \; r_1 = r_2$.  We use the name
$\rtkypr$ for this defined relation:
\begin{align*}
\rtkypr\; r_1 \; r_2 
\quad \text{if and only if} \quad
\rtpubof \; r_1 = r_2 .
\end{align*}

The following relation $\rtinvop$ is not a runtime primitive, it is a
definable relation convenient for analysis.

\begin{definition}
 The \emph{runtime inverse} $\rtinvop$ is another definable relation,
given by
\begin{align*}
  \rtinv{r_1}{r_2} \quad &\text{if $(r_1,r_2)$ or $(r_2, r_1)$  make a
                       runtime key pair} \\
  \rtinv{r}{r} \quad &\text{if $r$ is not of sort \Akey or \Ikey}
\end{align*}
\end{definition}
It is easy to see that the relation \rtinvop is actually a function:
for each $r$ there is a unique $r'$ such that
$\rtinv{r}{r' }$.    But this function is not feasibly
implementable, under the assumption that one cannot feasibly compute
the private part of a key pair from the public part.  So we prefer to
make our core definitions and axioms below in terms of \rtinvop as a
relation: it is clearly \emph{is} feasibly implementable under the
assumption that \rtpubof is.

\subsection{The Axioms}\label{sec:axioms-about-runtime}

\paragraph{Pairing Axiom}
The operations \rtfrst and \rtscnd are the usual projections
characterizing pairs.
  \begin{align}
    \rtpair\ r_1 \ r_2 = r 
    &\; \leftrightarrow \; \rtfrst r \opteq  r_1 \land \rtscnd r \opteq r_2 
      \label{eq:pair-project}
  \end{align}

\paragraph{Axiom about \rtgen}
The \rtgen operation delivers values of the
appropriate sorts.
\begin{align} 
  \rtsrt (\rtgen \ n \ srt) &= srt  \label{eq:gen-sort} 
\end{align}

\paragraph{Axioms about \rtpubof}  
The \rtpubof partial function makes a bijection from sort \Ikey to the sort
\Akey (and is undefined off of the sort \Ikey).
\begin{align} 
\label{eq:pub-of-sorts}
  \rtpubof r_1 =  r_2 &\to
                        \sortof(r_1) = \Ikey \land
                        \sortof(r_2) = \Akey    
  \\
  \label{eq:bijection1}
  \sortof(r_1) = \Ikey &\to
                         \exists ! \,r_2,
                         \rtpubof r_1 =  r_2
  \\
  \label{eq:bijection2}
  \sortof(r_2) = \Akey &\to
                         \exists ! \,r_1,
                         \rtpubof r_1 = r_2
\end{align}

\paragraph{Encryption Axiom}
The 
standard relationship between encryption and decryption is of course
\[
 \rtencr\ r_p \ r_{ke} \ r_e \ 
    \leftrightarrow \rtdecr \ r_e \ \inv{(r_{ke})} = r_p 
\]
when expressed using the runtime inverse function $\inv{(-)}$.
Since the inverse function is not feasibly computable we prefer
the following formulation in 
terms of \rtinvop.
\begin{align}
  \label{eq:encr-decr}
  \rtencr\ r_p \ r_{ke} \ r_e  \land \rtinv{r_{ke}}{r_{kd}}
  &\to \rtdecr \ r_e \ r_{kd} = r_p 
  \\
  \label{eq:decr-encr}
  \rtdecr \ r_e \ r_{kd} = r_p \land \rtinv{r_{ke}}{r_{kd}}
  &\to
  \rtencr\ r_p \ r_{ke} \ r_e  
\end{align}

\paragraph{Summary: the Runtime Theory}

\begin{definition}
  The  theory \rtThy\ comprises the axioms %
  (\ref{eq:pair-project}),
  (\ref{eq:gen-sort}), 
  (\ref{eq:pub-of-sorts}),
  (\ref{eq:bijection1}), 
  (\ref{eq:bijection2}), 
  (\ref{eq:encr-decr}), and
  (\ref{eq:decr-encr}) .
\end{definition}

In the Coq code for \molly we use a typeclass to capture the theory \rtThy.

\section{Compilation}
\label{sec:compilation}

Here we outline the structure of the compilation process.

The notion of ``saturation'' of a \proc plays a key role, and we
present a  careful discussion of saturation \refsec{sec:saturation}.
But
for an intuition, imagine generating code to process a reception of a
term $\Prr{\Dt 1}{\Dt 2}$.
We will bind $\Prr{\Dt 1}{\Dt 2}$ to a new location $\loc$, but will
also want to generate code that serves to ensure that an incoming
value at runtime is of the right form.    We do this bind emitting
code for a bonding of $t_1$ to another location $\loc_1$ with the
constraint that $\loc_1$ is equal to $\Frst{\loc}$.
That (plus the corresponding process for $t_2$) serves to check that
and incoming value really is a pair, since otherwise the $\Frst{\loc}$
(or the $\Scnd{\loc}$ will fail at runtime).
Then we need code to ensure that the value corresponding to $t_1$
really is of type $\Data$ \dots and so forth.

Saturation is the process of emitting all the bindings and checks
required for a \proc to be closed under these obligations.

\subsection{The Main Loop}
\label{sec:main-loop}

The overall structure of the compiler is simple. %
Given an input role \rl, we
\begin{enumerate}
\item initialize a \proc \pr with generated values, as explained below
\item then loop through the actions of the role:
  \begin{enumerate}
  \item if the current action is an \Prmop or a reception \Rcvop of a term $t$ we %
    add an Event statement to \pr recording the input or reception; %
    then add a statement binding $t$ to a new location;
    then saturate \pr

  \item if the current action is an \Retop or a \Sndop of a term $t$ we %
    add an Event statement to \pr recording the output or transmission; %
    then saturate \pr
  \end{enumerate}
\item since saturation can fail, the return type of each of 
the above steps returns, and of the compilation as a whole, is
$\option{\proc}$.
\end{enumerate}

The state of the compilation at any point is a record with the
following data
\begin{itemize}
\item the role to be compiled
\item the current \proc 
\item a list $\mathit{done}$ of the role actions treated so far
\item a list $\mathit{todo}$ of the role actions yet to be treated
\end{itemize}

\subsubsection{Invariants}
\label{sec:invariants}
To express our invariants we first note that 
the Bind statements of any \proc  naturally build a relation \tlop from terms
to locations:
\begin{definition} \label{def:tl}
  The relation $\tlop$ from terms to locations is defined as
  \[
  \tl{t}{\loc} \text{ if for some $e$, }
  (\Bind {t}{l} {e}) \text{ is in \pr .}
\]
\end{definition}

Using \tlop we define the following invariants maintained by the
compilation process.

\begin{enumerate}
\item The concatenation of $\mathit{done}$ with $\mathit{todo}$ is the original
  role
\item The relation \tlop systematically relates 
the list of terms $\mathit{done}$ and the trace of \pr.
More precisely we have 
\[ \mapR  \actmap{\tlop} \; \mathit{done} \; \pr
\]
\item  \pr is saturated.
\end{enumerate}

These invariants lead to the following properties of a successful
compilation returning a \proc \pr.

\begin{itemize}
\item From invariants 1 and 2:  the role \rl and the trace of \pr are related as
\[ \mapR  \actmap{\tlop} \; \rl \; \trace{\pr}
\]
\item \pr is saturated
\end{itemize}

Those closure properties are key for ensuring our core correctness
property, the Reflecting Transcripts property.
In \refsec{sec:saturation} we explain those properties and the
mechanisms to ensure them; in
\refsec{sec:relating} we present the proof of 
the Reflecting Transcripts theorem.

\subsection{Initialization}
\label{sec:generation}

Note the difference in character between the term $\Dt{3}$ and the
terms $\Dt{1}$ and $\Dt{2}$ in the role below.
  \begin{align*}
    & ... \\
    &\Rcv {\Ch 1} {\En {\Dt 1} {\Hs {\Dt 2}}} ;\\
    &... \\
    &\Snd {\Ch 1} {\Dt 3}\\
    & ...
  \end{align*}
  Assuming there are no occurrences of $\Dt 3$ prior to the one shown,
  the term $\Dt{3}$ is to be freely chosen by the agent executing the
  role. 

  We assume that the runtime has a mechanism for constructing values of
  a given sort.   Correspondingly our \procs have an expression \Genrop
  intended to be implemented by this mechanism.

  We next explain how the compiler emits the necessary statements.
  Key pairs require some special treatment, which we explain after
  giving the routine case.

\subsubsection{Generation Bindings}
\label{sec:generation-bindings}
  \begin{definition}[Polarity]
    Let \rl be a role and let $t$ be an elementary term occurring as a
    subterm in \rl.   %
    
    We say that has \emph{negative polarity} in \rl
    if the first event of \rl in which $t$ is a subterm is either a
    \Paramop or a \Rcvop.

    We say that has \emph{positive polarity} in \rl
    if the first event of \rl in which $t$ is a subterm is either a
    \Retop or a \Sndop.
  \end{definition}
  
  For each elementary term $t$ with positive polarity which is not an
  asymmetric key, the compiler emits code binding $t$ to a location
  with a corresponding \Genrop expression. Specifically, we emit the
  statement
    \begin{align}\label{eg:1}
    \Bind{t}{\loc} {\Genr{s}{k}}
  \end{align}
  where $\loc$ is a fresh location, %
  $s$ is the sort of $t$, and %
  $k$ is an natural-number index that identifies which occurrence of
  $t$ is being treated.

There is a subtlety, though, about key pairs for asymmetric
encryption, which we explain next.

\paragraph{Generation and Key Pairs}
\label{sec:generation-key-pairs}

Recall the definition of symbolic key pair from \ref{sec:terms}.
It will be convenient to use the shorthand %
$\pri{n}$ and
$\pub{n}$ for
$\Ik {\Av {n}}$ and
$ \Ak {\Av {n}}$, respectively.

If both \pri{n} and \pub{n} have positive polarity, we should not
generate calls to \Genrop independently for the two terms; this would
lose the constraint that their values should be key pairs.  And if exactly
one of them has positive polarity then any occurrence of its key
partner is to be received later.  The \proc will need code to check
that this reception is suitable, that is, its value must make a
runtime key pair with the value of the positive-polarity original
term.  This suggests that the symbolic key \emph{partner} of the term
with positive polarity must be associated with a binding.

The following process implements these ideas.

\begin{quote}
  \emph{If either \pri{n} or \pub{n} has positive polarity:} 
  add the two bindings
  \begin{align} 
    \label{eq:bind-pri}
    \Bind{\pri{n}}{\loc_{pri}} {\Genr{\Ikey}{ k}} 
    \\
    \label{eg:bind-pub}
    \Bind{\pub{n}}{\loc_{pub}} {\PubOf{\loc_{pri}}}
  \end{align}
  where $\loc_{pri}$ 
  and  $\loc_{pub}$ are fresh locations and $k$ is an index.
\end{quote}

If neither \pri{n} nor \pub{n} has positive polarity,  nothing
  happens for these terms during initialization: saturation will
  ensure that \pr have appropriate \CKyprop checks to ensure that
  locations associated with these terms make a key pair.

Now, the semantics of \procs will ensure that the values associated with
$\loc_{pri}$ and $\loc_{pub}$ will be runtime key pairs, because
of the $\PubOf{\loc_{pri}}$ expression.
But 
we should also check that when we generate a key pair then our \proc has
other runtime checks that we expect.  For instance let us suppose first
that \pri{n} occurs with positive polarity and that \pub{n} occurs
with negative polarity, say by a binding
\[ \Bind {\pub{n}} {\loc_1} {e} .
\] 
We will want to know that \pr has sufficient checks to ensure that 
 the runtime values associated with
$\loc_{pri}$ and $\loc_{1}$ will make a key pair.

By our initialization process \pr also has
\begin{align}
  \label{eg:3}
  \Bind{\pub{n}}{\loc_{pub}} {\PubOf{\loc_{pri}}}
\end{align}
But as a regular part of our \proc construction we ensure that whenever
two locations are assigned to the same elementary term we emit a 
assertions implying that the two locations have the same value at
runtime.
This is the Check Key Pair Condition of Definition~\ref{def:closed}.
Specifically, in our current scenario we will emit assertions 
implying that $\loc_{pub}$ and $\loc_{1}$ will have the same runtime
values.
This, in concert with the binding~\ref{eg:3},  ensures that  the values associated with
$\loc_{pri}$ and $\loc_{1}$ will be runtime key pairs.

A similar argument shows that if %
 \pub{n} occurs with
positive polarity and that \pri{n} occurs with negative polarity, 
then our \proc will have checks sufficient to ensure that any runtime
values assigned to the corresponding variables will be runtime key pairs.

\subsubsection{Generation and Initialization}
\label{sec:gener-init-}
  
  The current compiler emits all code for generating values at the
  initialization phase.  There is no necessity for doing it eagerly in this way:
  one could recognize generating terms ``on the fly'' when processing
  a send or return and do generation then.  But it is convenient
  to do these Generations in the initialization phase since it makes
  analysis and proof of correctness a bit smoother.

\subsection{The Structure of Expressions}
\label{sec:struct-expr}

Our \proc expressions are not a recursive type per se: expressions are
not arguments to expressions. But the binding
structure in a \proc gives an implicit recursive structure on expressions.  For
example, if $e$ is the expression $\Pair{\loc_1}{\loc_2}$ and our
\proc has bindings $\Bind {t_1}{\loc_1}{e_1}$ and
$\Bind {t_2}{\loc_2}{e_2}$ then it is natural to think of $e$ as being
a pairing of $e_1$ with $e_2$.  

Note there are exactly four expression operators
that do not take locations as arguments:
\[
\Paramop, \; \Readop, \; \Quotop \; \text{and} \; \Genrop
\]
and these take only indices serving to distinguish  occurrences or to
name a sort or a string.

So every expression built by the compiler can be thought of %
(modulo the indirection pointed out above involving locations) %
as a tree built by the other operators, whose leaves are essentially
expressions built from %
simple $\Paramop,$ $\Readop,$ $\Quotop,$ {and} $\Genrop$ expressions.

In Section~\ref{sec:struct-expr-binding} we show how this leads to a
useful perspective on the overall structure of a \proc.

\section{Saturation}
\label{sec:saturation}

In the previous section we described  saturation informally;
here we define it carefully.   First we introduce a convenient
relation on locations.
The \proc statement %
$\CSame{\loc_1}{\loc_2}$ compares the values in 2 given locations.  We
will need to work with the equivalence relation on locations this generates.

\begin{definition}[Sameness] \label{def:sameness}
  The relation $\sameness{\!}{\!}$ is the least equivalence relation on locations such
  that $ \sameness{\loc_1}{\loc_2}$  whenever the statement
  $\CSame{\loc_1}{\loc_2}$ is in \pr.
\end{definition}

\subsection{Saturated Procs}
\label{sec:saturated-procs}

We define two desirable conditions for a \proc: that it be closed and
justified.  To be saturated (Definition~\ref{def:saturation}) is to
satisfy each of these.

\begin{notation} \label{def-notation}
To avoid verbosity we will use, in 
Definitions~\ref{def:closed} and \ref{def:justified},
the convention that writing
\begin{align*}
  & \Bind {t} {\loc}{e} \\
   \text{is shorthand for} \\
  & \Bind {t} {\loc}{e} \text{ is one of the statements in the
  procedure \pr}.
\end{align*}
\end{notation}

\subsubsection{Closure}
\label{sec:closure}

Closure is the process of (i) adding binding  statements to a \proc
in order to reflect the information known about parameters and
messages received or generated and (ii) adding checks to reflect the
constraints---such as sameness---among locations.

We first define the condition of ``being closed'' and then, in 
\refsec{sec:rules}, give an algorithm for achieving closure.

\subsubsection{Motivation for the Closure Conditions}
\label{sec:motiv-clos-cond}
\begin{itemize}
\item Pair Elimination: a \proc satisfying this axiom is guaranteed to
  have  statements
  in deconstructing a received value and storing the values derived
  for future use.
\item Decryption: as for Pair Elimination, except that this axiom only
  ensures that we deconstruct a received encryption when the
  appropriate decryption key is available.

\item Pair, Encryption, and Hash Introduction:
  these axioms ensure that we have statements deriving new values to be used for
  transmission {and} for potential use in constructing decryption
  keys or hash-values when received hashes are to be checked.

\item Check Hash: since there is no operator to deconstruct a hash, we
  can only check the suitability of a value received when 
  a hash is
  expected by comparing it to a known hash of that value.

\item Check Same: ensures that we have code for the assertion that
  different locations bound to the same elementary term do indeed hold
  the same value.

\item Check Sort: similarly to Check Same we need code for the
  assertion that a value received has the expected sort.


\item Check Key Pair: is self-explanatory

\item Check Quote: is self-explanatory
\end{itemize}

We need the following taxonomy on expressions for the definition and
analysis of the closure conditions.   The intuition for these
definitions is that they capture circumstances where we do not need
to do projections (out of pairs) or decryptions (of encryptions)
during saturation.

\begin{definition}
Let \pr be a proc and $e$ and expression.
Then $e$ is
  \begin{itemize}
  \item a \emph{pair expression for} %
    term $\Pr{t_1}{t_2}$ in proc \pr if $e$ is of the form
    $\Pair{l_1}{l_2}$ and there are bindings in \pr of the form
    $\Bind{t_1}{l_1}{e_1}$and
    $\Bind{t_2}{l_2}{e_2}$.

  \item an \emph{encryption expression for} %
    term $\En{t_1}{t_2}$ in proc \pr if $e$ is of the form
    $\Encr{l_1}{l_2}$ and there are bindings in \pr of the form
    $\Bind{t_1}{l_1}{e_1}$and
    $\Bind{t_2}{l_2}{e_2}$.

\end{itemize}
\end{definition}

\begin{definition}[Closed Proc]
  \label{def:closed}
  Let \unv be a set of terms. %
  A \proc \pr is \emph{closed} if it satisfies the universal closures of
  the following formulas..
    \begin{align*}
        {\textbf{Pair Introduction Condition}} \\
        \Bind {t_1}{\loc_1}{e_1} 
        \land\
        \Bind {t_2}{\loc_2}{e_2} 
        \land\
        \Prr{t_1}{t_2}\in \unv
        &\to 
        \\
        \exists \loc,
        \Bind {\Prr{t_1}{t_2}}  {\loc} {\Pair{\loc_1}{\loc_2}} 
        \\
        \\
        {\textbf{Encryption Introduction Condition}} \\
        \Bind {t_1}{\loc_1}{e_1} 
        \land\
        \Bind {t_2}{\loc_2}{e_2} 
        \land\
        \En{t_1}{t_2}\in \unv
        &\to 
        \\
        \exists \loc,
        \Bind {\En{t_1}{t_2}}  {\loc} {\Encr{\loc_1}{\loc_2}} 
        \\ 
        \\
        \textbf{Hash Introduction Condition} \\
        \Bind {t_1}{\loc_1}{e_1} 
        \Hs{t_1} \in \unv
        &\to \\
        \exists \loc,
        \Bind {\Hs{t_1}}  {\loc} {\Hash{\loc_1}} 
        \\
        \\
        \textbf{Pair Elimination Conditions} \\
        \Bind {\Prr{t_1}{t_2}}  {\loc} {e} 
        & \ \land\ \\
        \text{$e$ not a pair expression for $\Pr{t_1}{t_2}$}
        & \to \\
        ( \exists \loc_1,
        \Bind {t_1}{\loc_1}{\Frst{\loc}} \
        \land \\
        \exists \loc_2,
        \Bind {t_2}{\loc_2}{\Scnd{\loc}} \; )
        \\
        \\
      \textbf{Decryption Condition} \\
      \Bind{\En{p}{k}} {\loc} {e}  
      & \ \land\ \\
      \text{$e$ not an encryption expression for $\En{t_1}{t_2}$}
      & \ \land\ \\
      \Bind{\inv{k}}{\loc_1}{e_1} 
        &\to \\
        \Bind{p}{\locnew}{\Decr{\loc}{\loc_1}}
        \\ \\
        \textbf{Check Hash Condition}  \\
        \Bind {\Hs{t}} {\loc_h} {e_h} 
        \land \Bind {t} {\loc_t} {e_t} 
        & \to \\
        \CHash{\loc_h}{\loc_t}
        \\
        \\
        \textbf{Check Equality Condition}  \\
        \elementary\ t
        \land\
        \Bind {t}{\loc_1}{e_1} 
        \land
        \Bind {t}{\loc_2}{e_2} 
        &\to \\
        \sameness{\loc_1}{\loc_2}
        \\
        \\
        \textbf{Check Quote Condition} \\
        \Bind {\Qt{s}}{\loc}{e}  
          &\to \\
          \CQot{\loc}{s}
          \\
          \\
        \textbf{Check Sort Condition}  \\
        \elementary\ t
        \land\
        \Bind {t}{\loc}{e} 
        &\to \\
        \exists \loc_1,
        \sameness{\loc}{\loc_1} 
        \land\
        \CSrt{\loc}{(\sortof t)} 
        \\
        \\
        \textbf{Check Key Pair Condition} \\
        ( t_1 , t_2) \text{ make a symbolic key pair } 
        \land\ \\
        \Bind{t_1}{\loc_1}{e_1}
        \land\
        \Bind{t_2}{\loc_2}{e_2}
        &\to \\ 
        \exists \loc'_1 \; \loc'_2 \; e'_1 \; e'_2 ,
        \Bind{t_1}{\loc'_1}{e'_1} 
        \land\
        \Bind{t_2}{\loc'_2}{e'_2} \
        \land\
        \\
        \sameness{\loc_1}{\loc'_{1}},
        \land\ 
        \sameness{\loc_2}{\loc'_{2}}
        \land\
        \CKypr{\loc'_1}{\loc'_2} \quad
    \end{align*}

\end{definition}

\subsubsection{Being Justified}
\label{sec:being-justified}
  A \proc \pr is ``justified'' if, intuitively
  \begin{itemize}
  \item received encryptions always have decryption keys available,
    and
  \item whenever $\Hs{t}$ is bound in \pr then $t$ is also bound
  \end{itemize}
  Formally (we continue to employ Notation~\ref{def-notation}):

\begin{definition}[Justified Proc]
\label{def:justified}
  A \proc \pr is \emph{justified} if it satisfies
  \begin{align*}
      \textbf{Encryption Justification} \\
      \Bind {\En{p}{k}} {\loc} {e}
      \land\ \text{non-Encryption } e
      &\to \\
      \exists \loc_1, \exists e_1, \Bind {\inv{k}}{\loc_1}{e_1} \;
      \\ \\
      \textbf{Hash Justification} \\
      \Bind {\Hs{t_1}}  {\loc} {e} 
      \land\ \text{non-Hash }\ e
      &\to \\
      \exists \loc_1, \exists e_1, \Bind {t_1}{\loc_1}{e_1} 
  \end{align*}
\end{definition}

Being justified is not a property that we can ensure of the \procs the
compiler builds.  It is ultimately a property of the role we are
compiling: it will fail if the parameters and expected receptions of
the role do not provide the material needed to construct needed
decryption keys or bodies of hashes.

\subsubsection{Saturated}
\label{sec:saturated}

\begin{definition}[Saturated Proc]
\label{def:saturation}
  A \proc \pr is \emph{saturated} with respect to a set of terms \unv
  if it is closed with respect to \unv %
  and is  justified.
\end{definition}

We will be interested in the case where \unv is the set of all
subterms of terms occurring in a role to be compiled.

\subsection{The Saturation Process}
\label{sec:rules}

In this section we present an algorithm for taking a \proc \pr and
returning a saturated extension of \pr (or failing).

\subsubsection{Motivation for the Closure Rules}
\label{sec:motiv-clos-rules}
The closure rules that follow are recipes for transformations that
 take an arbitrary \proc \pr (say, the state of a \proc
immediately after processing an event) and ultimately return a \proc $\pr_{*}$ that
satisfies the conditions of \refdef{def:closed}.  

The rules we present here do precisely this, in the sense that
is a  \proc $\pr_*$ is a fixed point with respect to these rules then 
$\pr_*$ satisfies the 
the conditions of \refdef{def:closed}.  This is the content of %
\refthm{thm:closed-closed}.

For the most part
the axioms of \refdef{def:closed} can be read ``operationally'' in the
sense that they are Horn sentences which can be made true by
augmenting the \proc so that it satisfies the consequent whenever the
antecedent is found to be true.  But there are some small obstacles to
interpreting the conditions of \refdef{def:closed} naively,
specifically (i) 
witnessing the existential quantifiers in the consequents, and (ii) the
logistics of ensuring that certain statements involving global
relations such as $\samenessop$ are satisfied.

The correspondence between 
the rules below  and 
the axiomatic conditions %
Pair Elimination, %
Decryption, %
Pair Introduction, Encryption Introduction, and Hash Introduction %
and is clear: for a given
instance of the antecedent of one of these conditions %
one need only (if necessary) generate a
fresh location to witness the existential quantifier and generate 
appropriate \Bindop statements to establish the consequent.
Thus several of the rules use the reference \locnew : this is a reference to a
variable not occurring elsewhere in the \proc.

For the assertions %
Check Same, Check Sort, Check Hash, Check Quote, and Check Key Pair,
the subtlety is that we want to ensure relationships between locations
based on the sameness relation $\samenessop$ without generating an
excessive number of $\CSameop$ statements.  %
We do this by the trick of identifying and exploiting the ``first
location'' for a term $t$: this will be the least location $l$ such
that a statement $\Bind{t}{l}{e}$ occurs in \pr.
In fact there is nothing special
about choosing the \emph{least} location; all that matters is that
there we identify one distinguished location for each elementary term
bound in \pr.

Once the rules have been identified, we saturate the \proc 
arbitrarily by the rules.
\refsec{sec:receptions} explains why 
we cannot simply do a 
 syntax-directed
application of the inference rules.





  



\subsubsection{The Closure Rules}
\label{sec:closure-rules}

The closure rules are defined in the context of a set \unv of symbolic terms;
in practice  \unv  will be the
set of subterms occurring in the role being compiled.

\begin{definition}[Closure Rules]
  \label{def:rules}

  Fix a universe \unv of symbolic terms. 

  The eleven inference rules for \emph{closing} a \proc \pr are the following.

  \paragraph{Pair Introduction Rule} \hfill

  Here we require that $\Prr{t_1}{t_2}$ be in \unv, and there are no
  bindings for $\Prr{t_1}{t_2}$ in \pr .

  \begin{prooftree}
    \AxiomC {$ \Bind{t_1}{\loc_1}{e_1} $}
    \AxiomC {$ \Bind{t_2}{\loc_2}{e_2} $}
    \BinaryInfC {$ \Bind{\Prr{t_1}{t_2}} {\locnew} {\Pair{\loc_1}{ \loc_2}} $}
  \end{prooftree}

  \paragraph{Encryption Introduction Rule}
  \hfill

  Here we require that %
  $\En{p}{k}$ be in \unv, and there are no  bindings for $\En{p}{k}$ in \pr.

  \begin{prooftree}
    \AxiomC {$ \Bind{p}{\loc_1}{e_1} $}
    \AxiomC {$ \Bind{k}{\loc_2}{e_2} $}
    \BinaryInfC {$ \Bind{\Prr{p}{k}} {\locnew} {\Pair{\loc_1}{ \loc_2}} $}
  \end{prooftree}

  \paragraph{Hash Introduction Rule}
  \hfill

  Here we require that $\Hs{t}$ be in \unv, and there are no  bindings for $\Hs{t}$ in \pr.

  \begin{prooftree}
    \AxiomC {$ \Bind{t}{\loc}{e} $}
    \UnaryInfC {$ \Bind {\Hs{t}} {\locnew} {\Hash{\loc}} $}
  \end{prooftree}

  \paragraph{Pair Elimination Left and Right Rules}
  \hfill

  We apply these rules when
  $e$ not a pair expression for $\Pr{t_1}{t_2}$

  \begin{minipage}{0.5\linewidth}
    \begin{prooftree}
      \AxiomC {$
        \Bind{\Prr{t_1}{t_2}}{\loc}{e}
        $}
      \UnaryInfC {$
        \Bind{t_1}{\locnew}{\Frst{\loc}}
        $}
    \end{prooftree}
  \end{minipage}
  \begin{minipage}{0.5\linewidth}
    \begin{prooftree}
      \AxiomC {$
        \Bind{\Prr{t_1}{t_2}}{\loc}{e}
        $}
      \UnaryInfC {$
        \Bind{t_2}{\locnew}{\Scnd{\loc}}
        $}
    \end{prooftree}
  \end{minipage} \\

  \paragraph{Decryption Rule}
  \hfill

  We apply this rule when
  $e$ not an encryption expression for $\En{p}{k}$
  The active premise is the encryption binding.

  \begin{prooftree}
    \AxiomC {$
      \Bind{\En{p}{k}} {\loc} {e} \quad 
      \Bind{\inv{k}}{\loc_1}{e_1} 
      $}
    \UnaryInfC {$ \Bind{p}{\locnew}{\Decr{\loc}{\loc_1}} $}
  \end{prooftree}

  \paragraph{Check Hash Rule}
  \hfill

  The active premise is $ \Bind{\Hs{t}}{\loc_h}{e_h}$
  \begin{prooftree}
    \AxiomC {$ \Bind{\Hs{t}}{\loc_h}{e_h}$}
    \AxiomC {$ \Bind{{t}}{\loc_t}{e_t}$}
    \BinaryInfC {$\CHash{\loc_h}{\loc_t} $} 
  \end{prooftree}


           %


  \paragraph{Check Quote Rule}
  \hfill

  
\begin{prooftree}
    \AxiomC {$ \Bind{\Qt{s}}{\loc}{e}$}
    \UnaryInfC {$\CQot{\loc}{s} $} 
  \end{prooftree}

  \paragraph{Check Sort Rule}
  \hfill

  We apply this rule when 
  $t$ is an elementary term
  and $\loc_{1}$ is the first location for $t$ in \pr.

  \begin{prooftree}
    \AxiomC {$ \Bind{t}{\loc}{e}$}
    \UnaryInfC {$\CSrt{\loc}{\sortof t} $}
  \end{prooftree}

  \paragraph{Check Same Rule}
  \hfill

  We apply this rule when %
  $t$ is an elementary term, $\loc_{f} < \loc_{1}$, and $\loc_{f}$ is the first
  location for $t$ in \pr.

  The active premise is the binding whose location is $\loc_{1}$.

  \begin{prooftree}
    \AxiomC {$ \Bind{t}{\loc_1}{e_1}$}
    \AxiomC {$ \Bind{t}{\loc_f}{e_f}$}
    \BinaryInfC {$\CSame{\loc_1}{\loc_f} $}
  \end{prooftree}

  \paragraph{Check Key Pair Rule}
  \hfill

  We apply this rule when  $(t_1, t_2)$ makes a symbolic key pair, 
  $\loc_{1}$ is the earliest location for $t_1$, and
  $\loc_{2}$ is the earliest location for $t_2$.
  
  The active premise is the binding whose location is $\loc_{2}$.
  \begin{prooftree}
    \AxiomC {$ \Bind{t_1}{\loc_1}{e_1}$}
    \AxiomC {$ \Bind{t_2}{\loc_2}{e_2}$}
    \BinaryInfC {$\CKypr{\loc_1}{\loc_2} $}
  \end{prooftree}
\end{definition}

\subsubsection{Closure is not Syntax-Directed}
\label{sec:receptions}

As the examples below will demonstrate, when we use the rules to
construct a saturated \proc, we cannot apply them in a naively
syntax-directed way. %
In the course of analyzing a reception we sometimes must use bindings
that we can access but which have not yet themselves been fully
analyzed.

\begin{example}
  Suppose we receive the pair
  \[
    (\En{b}{k}, \; \En{k}{\En{b}{k}})
  \]
  For readability we temporarily revert to ordinary pair notation
  instead of writing
  \[
    \Prr{\En{b}{k}} {\En{k}{\En{b}{k}}}
  \]
  Assuming $b$ and $k$ are elementary, 
  we can successfully analyze this by 
  (i) using %
  the first component of the pair as decryption key for the second
  component, thereby obtaining $k$, then (ii) using 
  $k$ to decrypt the first component. %
  The net result is that will generate bindings for each of 
  \[ \{ \; 
    (\En{b}{k}, \; \En{k}{\En{b}{k}}), 
    (\En{b}{k}), \; 
    \En{k}{\En{b}{k}}, \; {k} , \; {b}
    \; \}
  \]
\qed \end{example}

We can also vary the example so that it does not depend on randomized
operations being used as keys.
\begin{example}
  Suppose we receive the pair
  \[
    (\Hs{(b,k)}, \; \En{(b,k)}{\Hs{(b,k)}})
  \]
  Assuming $b$ and $k$ are elementary, we can successfully analyze this
  by strategy similar to the previous one.
\qed \end{example}
The next example shows we may need to apply
construction (i.e.~introduction) rules in the course generating code
for a reception.
\begin{example}
  Suppose we receive the pair
  \[
    (\Hs{(b,k)}, \; \En{(b,k)}{(\Hs{(b,k)},\Hs{(b,k)})})
  \]
  Assuming $b$ and $k$ are elementary, we can successfully analyze this
  by %
  using the first component to 
  construct the pair
  ${({\Hs{(b,k)}},{\Hs{(b,k)}})}$, so that it can be used as a
  decryption key for $\En{(b,k)}{({\Hs{(b,k)}},{\Hs{(b,k)}})}$.  
\qed \end{example}

So our code proceeds by (after somewhat arbitrarily ordering the
rules) applying the first rule that can fire and continuing until
reaching a fixed point.   We argue termination in \refsec{sec:termination}.

\subsubsection{Termination}
\label{sec:termination}

\begin{theorem}
  \label{thm:termination}
  Let \pr be a \proc and \unv a set of terms. %
  There are no infinite  sequences of saturation rules starting with
  \pr using \unv.
\end{theorem}
\begin{proof}
  The three introduction rules apply at most once for each $t \in \unv$,
  and the new bindings they add cannot be premises of any other rule.
  Thus there are at most  $|\unv|$ applications of introduction rules in
  any saturation process.

  So it suffices to argue that there can be only finitely many
  application of Checks and the  elimination rules %
  \textbf{Pair Elimination Left and Right and Decryption}.

  Let us say that a binding %
  $\Bind{t}{v}{e}$ is a \emph{redex} if it is the active premise of a
  rule whose conclusion is not in \pr.

  Note that each binding can be a redex for at most one rule, with two
  exceptions:
  $(\Bind {\Prr{t_1}{t_2}} {v} {e})$ can be a redex for both 
  Pair Elimination Left and Pair Elimination Right, and 
  a binding for an asymmetric key can be a  premise for Check Key Pair
  as well as for (one of) Check Sort or Check Same.

  Let us assign a \emph{weight} to each binding %
  $(\Bind{t}{v}{e})$ in \pr, by %
  (i) counting the number of rules for which it is an active redex and 
  (ii) multiplying this number by the size of $t$.

  For example, if $(\Bind {\Prr{t_1}{t_2}} {v} {e})$ is in \pr and neither
  the conclusion of Pair Elimination Left nor  Pair Elimination Right is
  in \pr then this binding gets weight $2 |\Prr{t_1}{t_2}|$.

  Then we say that the
  \emph{weight} of \pr is the sum of the weights of the bindings in \pr.
  We claim that each elimination or Check rule application decreases
  this weight.

  First: by inspection we see that when a rule fires, the active premise
  is no longer a premise for that rule.  

  Second: when a Check rule fires, the weight of \pr decreases by the
  size of term being bound.    No bindings are added by a Check rule.

  Finally, when Pair Elimination Left or Pair Elimination Right or
  Decryption fires, the size of the term  being bound is subtracted
  from the weight of \pr, and replaced by the weight of some term in a new
  binding. But this new term is smaller than the term in the redex.

  Thus the weight of the \proc decreases at each step, and saturation
  must terminate.
\end{proof}

\subsubsection{Correctness}
\label{sec:correctness}

The rules in \refdef{def:rules} that create new bindings (for example,
Pair Elimination) have an obvious relationship with the corresponding
declarative conditions in \refdef{def:closed}.  %
But the rules and conditions about check statements are more subtly
linked, essentially because the rules only add individual statements
to the \proc yet the condition on the check statements make reference
to more global properties of a \proc, in particular, the $\samenessop$\
relation.
The next theorem verifies that the rules are sufficient to enforce
the conditions we need.

\begin{theorem}
  \label{thm:closed-closed}
  Fix a set of terms \unv. %
  Suppose \pr is closed under the rules of \refdef{def:rules} (with
  \unv as bounding set). %
  Then \pr is closed (with respect to \unv) in the sense of \refdef{def:closed}.
\end{theorem}
\begin{proof} \hfill
  \begin{itemize}
  \item \emph{Verifying the} \textbf{Pair Introduction Condition}

    If %
    $\Bind {t_1}{\loc_1}{e_1}$ and %
    $\Bind {t_2}{\loc_2}{e_2}$ are in \pr %
    and
    $\Prr{t_1}{t_2}$ is in \unv %
    then the Pair Introduction Rule ensures that %
    there exists $\loc$ with
    $\Bind {\Prr{t_1}{t_2}}  {\loc} {\Pair{\loc_1}{\loc_2}}$ in \pr.

  \item \emph{Verifying the} \textbf{Encryption Introduction Condition}

    Just as for Pair Introduction, the Encryption Introduction Rule
    ensures this property directly.

  \item \emph{Verifying the} \emph{Verifying the} \textbf{Hash Introduction Condition}

    The Hash Introduction Rule ensures this property directly.

  \item \emph{Verifying the} \textbf{Pair Elimination Condition}

    The Pair Eliminations Rules ensure this property directly.

  \item \emph{Verifying the} \textbf{Decryption Condition}

    The Decryption Rule ensures this property directly.

  \item \emph{Verifying the} \textbf{Check Hash Condition}

    The Check Hash Rule ensures this property directly.




    
  \item \emph{Verifying the} \textbf{Check Quote Condition}

    The Check Quote Rule ensures this property directly.

  \item \emph{Verifying the} \textbf{Check Equality Condition}

    Suppose $t$ is elementary and %
    $      \Bind {t}{\loc_1}{e_1} $ and
    $      \Bind {t}{\loc_2}{e_2} $ are in \pr.
    
    Let $\loc_0$ be the first location for $t$ in
    \pr. %
    Since \pr is closed under the Check Same Rule,
    either 
    \begin{itemize}
    \item $\loc_0$ is $\loc_1$ and 
      $\CSame{\loc_1}{\loc_2}$ is in \pr, or
    \item $\loc_0$ is $\loc_2$ and 
      $\CSame{\loc_2}{\loc_1}$ in \pr, or
    \item $\loc_0$ is neither $\loc_1$ nor $\loc_2$ and 
      both
      $\CSame{\loc_0}{\loc_1}$  and
      $\CSame{\loc_0}{\loc_2}$ are in \pr
    \end{itemize}
    In all of these cases, %
    $ \sameness{\loc_1}{\loc_2}$.

  \item \emph{Verifying the} \textbf{Check Sort Condition}

    Suppose $t$ is elementary and %
    $ \Bind {t}{\loc}{e} $ is in \pr.
    
    Take $\loc_0$ to be the first location for $t$ in \pr.
    Since \pr is closed under the Check Sort Rule,
    $\CSrt {\loc_0}{(\sortof t)}$ is in \pr.
    By the just-proven fact that \pr satisfies the Check Equality
    Condition, $\sameness {\loc} {\loc_0}$  as desired.

  \item \emph{Verifying the} \textbf{Check Key Pair Condition}

    \sloppy Suppose $(t_1 , t_2)$ makes a symbolic key pair and  %
    $\Bind{t_1}{\loc_1}{e_1}$ and %
    $\Bind{t_2}{\loc_2}{e_2}$ are in \pr. %

    \fussy
    Take $\loc'_1$ and  $\loc'_2$ to be the respective first-locations
    for $t_1$ and $t_2$.
    Then certainly %
    $\sameness{\loc_1}{\loc'_{1}}$ and 
    $\sameness{\loc_2}{\loc'_{2}}$. %
    And    $\CKypr{\loc'_1}{\loc'_2}$ is in \pr by the
    Check Key Pair Rule.
    
  \end{itemize}
\end{proof}

\section{Some Results on Procs}
\label{sec:structure-procs}

\subsection{The Structure of Expressions in Bindings}
\label{sec:struct-expr-binding}

In Section~\ref{sec:struct-expr} we observed that expressions in
bindings can be thought of as the result of flattening of a recursive
type using locations.

We noted that we can view expressions in the presence of \proc
bindings as a tree whose leaves are expressions built from %
simple $\Paramop, \Readop, \Quotop, \text{and} \Genrop$ expressions.

Now we can say more.
The compiler makes initial bindings using $\Paramop$ and $\Readop$ when translating an input
or reception, respectively, and 
makes an initial bindings using
uses $\Pubofop$ and $\Genrop$ when 
initializing or preparing a transmission, respectively.
If we examine the subsequent saturation process 
we can see that 
\begin{itemize}
\item 
the expressions built when processing a parameter or a reception are precisely
those built 
\begin{itemize}
\item starting from \Paramop and \Readop and
\item using the ``destructive''
operators %
\[
    \Frstop \quad
    \Scndop \quad
    \Decrop
\]
\end{itemize}

\item the expressions built when initializing or processing a transmission or a return
  are precisely those built 
  \begin{itemize}
  \item starting from \Genrop, \Quotop, and \Pubofop and 
  \item using the ``constructive''
  operators
\[
    \Pairop \quad
    \Encrop \quad
    \Hashop \quad
\]
\end{itemize}
\end{itemize}

\subsection{Procs and Derivability}
\label{sec:derivability}

In this section we show the intimate connection between the \procs
constructed by {\molly} and Dolev-Yao derivability of symbolic terms.

Roughly speaking, the connection is this: if \pr is a \proc generated
from a role \rl by Initialization and our Closure Rules,  %
then the terms $t$ such that there is a binding
$ \Bind {t} {\loc} {e}$ in \pr are the terms that are Dolev-Yao
derivable from the input parameters and messages received in \rl.  %

Of course the $\Genrop$ operator allows us to bind \emph{any}
elementary term, so we have to exclude such generated terms.
Also the $\Qtop$ and $\Pubofop$ forms are not treated in
traditional Dolev-Yao so we need to enrich the system just slightly.

\begin{definition}[Enhanced Dolev-Yao]
  The  \emph{enhanced Dolev-Yao} inference system comprises the
  traditional 
  Dolev-Yao system with the addition of the following two rules
  \begin{enumerate}
  \item  derive  $\Qt{s}$ for any string $s$ 
  \item derive $\Ak {\Av {n}}$ from $\Ik {\Av {n}}$
  \end{enumerate}
\end{definition}
\begin{definition}[Obtained]
  A term $t$ is \emph{obtained} in role \rl if either 
$\Prm{t}$ is in \rl or $\Rcv{ch}{t}$ is in \rl for some $ch$
\end{definition}

Before proving the next two lemmas we make a simple observation.
The  rules in \refdef{def:rules} that create a \Bindop are %
  Pair Introduction,
  Encryption Introduction,
  Hash Introduction,
  Pair Elimination Left and Right,
  and Decryption.
  For each of these rules, 
  if we suppress everything but
  the symbolic terms in the hypotheses and conclusion, 
  we have an rule in the enhanced Dolev-Yao system.

\begin{proposition}
\label{prop:bind-then-derived}
  Suppose \pr is a \proc generated from a role \rl by
  our Initialization and  Closure Rules.  %
  
  If $ \Bind {t} {\loc} {e}$ is a statement in \pr such that the
  expression $e$ has no occurrence of the \Genrop operator,  %
  then $t$ is derivable in the enhanced Dolev-Yao system from the set of
  terms obtained by \rl.
\end{proposition}
\begin{proof}
  The proof is  an easy induction on the construction of \pr, using the
  observation immediately above.
\end{proof}

\begin{proposition}
  \label{prop:derived-then-bind}
  Suppose \pr is a \proc generated from a role \rl by
  our Initialization and  Closure Rules.  %
  Further assume that \pr is closed. %

  If term $t$ occurs as a subterm of \rl and is derivable in the
  enhanced Dolev-Yao system from the set of terms obtained by \rl,
  then there exist $\loc$ and $e$ such that %
  $ \Bind {t} {\loc} {e}$ is a statement in \pr.
\end{proposition}
\begin{proof}
  Let $\mathcal{D}$ be a shortest derivation of $t$
  from the set of terms obtained by \rl.
  The proof is by induction on \D,
 again using the fact that the
  enhanced Dolev-Yao rules are a kind of erasure of the
  \Bindop-creating closure rules.  %
  Since the closure rules have more structure than do Dolev-Yao
  inferences we need to check some
  details.   %
  
  Suppose the last inference in $\mathcal{D}$ is a Dolev-Yao
  Pair Introduction, say of a term $\Pr{t_1}{t_2}$.  %
  By the induction hypothesis \pr has  bindings for $t_1$ and for $t_2$.
  The  fact that $\mathcal{D}$ is a shortest
  derivation means that $\Pr{t_1}{t_2}$ has not already been derived,
  and so \pr has no bindings for  $\Pr{t_1}{t_2}$.  %
  That is, the side condition for the  Pair Introduction closure rule
  is satisfied, and since \pr is closed, the Pair Introduction closure
  rule will have fired, as desired.

  The arguments for Encryption Introduction and Hash Introduction are
  just the same.

  For Pair Elimination Left and Right and for Decryption the structure of
  the induction is similar but now as we apply the closure rules the
  side conditions are on the shape of the expressions $e$.

  Consider the case of the Decryption Rule.
  We have, in \D, an inference
  \begin{prooftree}
    \AxiomC {$
      {\En{p}{k}} \quad 
      {\inv{k}}
      $}
    \UnaryInfC {$ {p} $}
  \end{prooftree}
so we want to argue the \pr has a binding for $p$.
 By induction we have, in \pr, bindings of the form
  \[
      \Bind{\En{p}{k}} {\loc} {e} \quad \text{ and } \quad
      \Bind{\inv{k}}{\loc_1}{e_1} 
    \]
    Now we have two cases: either $e$ is an encryption expression for 
    $\En{p}{k} $ or not.

    If so, then \pr already contains a binding for $p$.  If not, then
    we may fire the Decryption closure rule.

    The argument for Pair Elimination is similar.
\end{proof}

\subsection{Executability}
\label{sec:executability}

Several authors (\eg,
\cite{caleiro2006semantics},\cite{basin2015alice},
\cite{chevalier2010compiling}) have defined notions of
\emph{executability} of a protocol, 
statically-checkable properties that give confidence that a protocol
can be run to completion.  %

 We will eventually define \emph{non-executability} as
\emph{failure of compilation}.
To motivate that we start with the question: how can compilation fail?

Compilation succeeds precisely when  saturation succeeds at each role event.
Closure always halts:  the process runs until no more rules can be
applied, and our termination analysis say this will eventually halt
(\refthm{thm:termination}) . %
So the only thing that can go wrong is that we halt with a 
closed \proc that
isn't justified.

So suppose the  \proc \pr is constructed from role \rl and is
 closed, but not justified.
This means
(cf. \refdef{def:justified}) that either
\begin{itemize}
\item there is a binding
      $\Bind {\En{p}{k}} {\loc} {e}$ in \pr, %
      with $e$ not an $\Encrop$-expression, %
      such that for no %
      $\loc_1, e_1$ do we have
      $\Bind {\inv{k}}{\loc_1}{e_1}$ in \pr,
      or
\item there is a binding
  $\Bind {\Hs{t_1}}  {\loc} {e}$ in \pr, %
  with $e$ not a $\Hashop$-expression, %
  such that for no 
  $\loc_1, e_1$ do we have
  $\Bind {t_1}{\loc_1}{e_1} $ in \pr.
\end{itemize}

Since in each case the terms are associated with a neutral expression,
the \proc needs to do a decryption or check-hash respectively.

But in the first case, the term $\En{p}{k}$ %
is derivable from the terms obtained in \rl %
(\refprop{prop:bind-then-derived}), but the term $\inv{k}$ is not
derivable from the terms obtained in \rl %
(\refprop{prop:derived-then-bind}).

Similarly, in the second case the term $\Hs{t_1}$ %
is derivable from the terms obtained in \rl, but the term $t_1$ is not
derivable from the terms obtained in \rl %

 So failure of compilation is reflected by the existence of terms from
the role that, by the results of the last section, cannot be Dolev-Yao
derived but that are required in order for the \proc to be able to
construct statements it needs.

This analysis motivates the following definition.
\begin{definition}
A role \rl is \emph{non-executable} if the process of 
Initialization followed by closure under the 
  rules of \refdef{def:rules} yields a \proc which is 
  not justified.
\end{definition}

In this situation our compilation process does not return a runnable
\proc but in a precise way, it is not the compiler that is to blame: either there
is an  encryption derivable from the role whose decryption key is not derivable, or
there is a derivable hash term whose body is not derivable.   Each of these 
situations is one in which the \proc at hand cannot successfully
validate a necessary check.

Summarizing informally: we call a role \emph{non-executable} if either %
some reception leads to an encryption whose decryption key cannot be
derived, or %
some reception leads to a hash whose body cannot be derived.
Non-executability can thus be determined statically: it is witnessed
at compile-time by a failure of saturation.

\paragraph{Executable vs Non-Executable}
It is a bit awkward that we have been exploring
\emph{non-executability} as opposed to \emph{executability.} 
But there is a good reason for that: to use the phrase
\emph{executable} as the negation of \emph{non-executable} is quite a
misleading choice of phrase!

To explain,  let us unravel 
what it means for a \proc \pr to be 
\textbf{not} {non-executable}.   This will mean that \pr is closed
and justified.   %
But \emph{execution of such a \pr can still fail.}

This is a necessary consequence of the fact that encryption can be
randomized.  
Different occurrences of $t$ in the protocol may be associated with
different locations in our \proc, let us say %
$\Bind {t}{\loc_1}{e_1}$ and
$\Bind {t}{\loc_2}{e_2}$.
And  if $t$ is a term containing a randomized encryption as a subterm,
then at runtime the locations $\loc_1$ and $\loc_2$ can have different
values.

 To see why this is a problem, return to the encryption case.
 Suppose \pr has
a binding $\Bind {\En{p}{k}} {\loc} e$  with $e$ neutral, so that the
term ${\En{p}{k}}$ was derived from a reception.
Since \pr is justified we know that we have
a binding for $\inv{k}$, say %
$\Bind {\inv{k} } {\loc_1} {e_1}$. 
 Saturation will have emitted  a suitable decryption statement %
$\Bind {p} {\loc_{new}} {\Decr {\loc}{\loc_1}}$.

Now imagine that the encryption above is a symmetric encryption,
so $\inv{k}$ is $k$.  But
if $k$ is a term itself containing a randomized encryption as a subterm, then
it is possible  at runtime  that %
\emph{the value of the occurrence of $k$ used to build the encryption is
  not the same as the value of the occurrence of $k$ used as
  decryption key}.  This means that decryption will fail at runtime.

Similarly (more simply, in fact) suppose $\Hs{t_1}$ and $t_1$ are each
bound, as required by being justified.   A Check Hash test can fail if
the two occurrences of $t$ have different runtime values.

 What's going on here is simply the fact that identity at the symbolic
 term level does not translate into identity of values as the
 bitstring level.  This is not a weakness of the compiler.  In a sense 
  it is a lack of expressiveness of our symbolic term language, an
 unavoidable consequence of the fact that the randomness of encryption
 is not reflected in the syntax of terms.

The takeaway from this discussion is that although we can statically
detect certain fatal obstacles to a role being able to be executed by
compiled code, the absence of those obstacles does not guarantee that the
compiled code will indeed run to completion.   And this is why we
hesitate to use the term ``executable.''

\section{Valuations, Stores, and Transcripts}
\label{sec:transcripts}

\subsection{Raw Transcripts}
\label{sec:raw-transcripts}

Transcripts are sequences of runtime actions describing the observable
behavior of a protocol execution.

\begin{definition}
A transcript is a list of \Act{\rtval}.
\end{definition}

\begin{example} \label{eg:xpt-resp1} 
The following is a transcript.
  \[ \Prm {r_1} ; \Rcv {r_1} {r_2} ; \Snd {r_1} {r_2} 
   \]
This describes an execution in which 
\begin{enumerate}
\item value $r_1$ is taken as a parameter;
\item a value $r_2$ is received at channel $r_1$;
\item that same  value $r_2$ is sent over channel $r_1$.
\qed
\end{enumerate}
\end{example}

Of course, it's not clear what this should mean if $r_1$ isn't in fact
a channel value.  And there will be more interesting structural
constraints to be observed if we expect a transcript to be executions
of roles or \procs.  In this section we define
\begin{enumerate}
\item what it means for a given transcript to denote a possible
  semantics of a given role \rl %
  (we will use the phrase ``valid transcript for \rl''), and 
\item 
what it means for a given transcript to denote a
possible execution of a given \proc \rl %
  (we will use the phrase ``valid transcript for \pr'').
\end{enumerate}
We use the phrase ``raw transcript'' to refer to a list
of $\Act{\rtval}$ without making any claims of validity.

\subsection{Transcripts for a Role}
\label{sec:transcripts-role}

To be a transcript for a role an $\Act{\rtval}$ sequence should
satisfy two conditions.

First, the number of items and their character should match the role:
for example, if the third element is the role is a \Sndop of a term
then the the third element of the transcript should be a \Sndop of a
runtime value, and so on.

Second, we expect a certain compositionality: %
for example if the role has $\Prr{t_1}{t_2}$ as a reception and then
${t_1}$ as a transmission, and the transcript %
associates runtime value $r$ with $\Prr{t_1}{t_2}$ and %
associates runtime value $r_1$ with $t_1$ then  $r_1$ should  be the
first projection of $r$.

The first constraint is straightforward to express.
The second constraint needs some attention, and is the subject of 
\refsec{sec:valuations}, culminating in 
Definition~\ref{def:valid-role-transcript}.

\begin{example}
Consider this role
  \begin{align*}
    [&\Prm {\Ch 1};  \ \Prm {\Ik {(\Av 2)}} ; \\
    &\Rcv {\Ch 1} {\En {\Nm 0} {(\Ak{ (\Av 2)} )}} ; \\
    &\Snd {\Ch 1}{(\Nm 0)}]
  \end{align*}
and consider the following raw transcript
  \begin{align*}
    &\Prm {r_1}; \ \Prm {r_2} ; \\
    &\Rcv {r_1} {r_3}; \\
    &\Snd {r_1} {r_4}
  \end{align*}

  By the way this is the first place where we see the technical
  utility of the \Actop parameterized data type: it is transparent how
  the actions over runtime values in the transcript are intended to
  correspond to the actions defined in the role.

  This transcript describes an execution in which %
  $r_1$ is a value assigned to channel 1; %
  $r_2$ is a value assigned to ${\Ik {\Av 2}}$ (the private key of the
  agent executing the role); %
  $r_3$ is the value received on $r_1$; and %
  $r_4$ is the value sent on $r_1$.

  This transcript satisfies our first informally-stated constraint.
  But the  more interesting second constraint requires (for example)
  \begin{enumerate}
  \item $r_1$ should be the sort of value that denotes a channel,
    $r_2$ should be a key, and so on.
  \item $r_3$ should be a value that arises by encrypting $r_4$
    with the key-partner of $r_3$
  \end{enumerate}
Being a ``valid'' transcript for a role will mean obeying sort
constraints (as in (1)  above) and
relationship between values (as in (2) above).   
\qed \end{example}

Here is an important point:   When a
term such as $\Ch{1}$ or $\Nm{0}$ is used more then once in a role, the
corresponding runtime value occurrences must be equal.
That is, there should be a \emph{function} from such terms to runtime values.
But since we are working with randomized encryption, the value
associated with a symbolic encryption %
$\En{p}{k}$ will not be uniquely determined by the values 
associated with ${p}$ and $k$.  Rather, there is a \emph{relation}
between symbolic terms that involve encryptions and runtime values.

This subtlety is reflected in \refdef{def:valuation}, where the
fundamental construct that generates valid transcripts is a relation; 
 this relation will be required to be functional on elementary terms.

\subsubsection{Valuation for a Role}
\label{sec:valuations}

Let \rl be a \role and let \tr be a transcript.

To say that the runtime events of \tr are related to \rl pointwise in
a systematic way is to say that there is a 
relation \rolvalop from $\Term$ to $\rtval$ such that
  \begin{align*}
    (\mapR {\actmap{\rolvalop}} \ \rl \tr)
  \end{align*}
  We may say that \tr is \emph{induced} by $\rolvalop$.


Next,
in order to view \tr as valid transcript for a role we will also 
 insist on some compositionality conditions, articulated in the next definition.
\begin{definition}[Valuation]
  \label{def:valuation}
  A relation   $\rolvalop \subseteq \Term \times \rtval$
  is a \emph{term valuation} if it satisfies the following conditions.
  \begin{enumerate}
  \item \rolvalop is functional on elementary terms:

    For all $t, r_1, r_2$, if $t$ is elementary and
    $\rolval{t}{r_1}$  and $\rolval{t}{r_2}$
    then $r_1 = r_2$

  \item $\rolvalop$ respects sorts : 

    For all $t, r$, if $t$ is elementary and 
    $\rolval{t}{r}$ then 
    $\sortof t = \rtsrt \ r$

  \item $\rolvalop$ respects pairing:

    For all $t_1, t_2, r$, %
    if $\rolval{\Prr{t_1}{t_2}}{r}$ then
    there exist $r_1, r_2$ such that
    $\rolval{t_1}{r_1}$, 
    $\rolval{t_2}{r_2}$, and
    $\rtpair{r_1}{r_2} = r$.

  \item  $\rolvalop$ respects hashing:

    For all $t_1, r$, %
    if $\rolval{\Hs{t_1}}{r}$ then
    there exist $r_1$ such that
    $\rolval{t_1}{r_1}$ and
    $\rthash{r_1} = r$.

  \item   $\rolvalop$ respects quote:

    For all $s, r$, %
    if $\rolval{\Qt{s}}{r}$ then %
    $ r = \rtquote s$

  \item $\rolvalop$ respects key pairs:

    For all $t_1, t_2, r_1, r_2$, %
    if $(t_1, t_2)$ makes a key pair,
    $\rolval{t_1}{r_1}$, and
    $\rolval{t_2}{r_2}$, then 
    $\rtkypr{r_1}{r_2} = true$

  \item $\rolvalop$ has the disjunctive encryption condition:

    For all $p, k_e, r_e$, %
    if $\rolval {\En{p}{k_e}} {r_e}$ then 
    at least one of the following holds:
    \begin{itemize}
    \item (the encr condition)
      there exists $r_p, r_{ke}$ such that 
      \begin{itemize}
      \item $\rolval {p}{r_p}$
      \item $\rolval {k_e}{r_{ke}}$
      \item $\rtencr {r_p}{r_k}{r_e}$
      \end{itemize}

    \item (the decr condition)
      there exists $r_p, r_{kd}$ such that 
      \begin{itemize}
      \item $\rolval {p}{r_p}$
      \item $\rolval {\inv{(k_e)}}{r_{kd}}$
      \item $\rtdecr {r_e}{r_{kd}} \opteq {r_p}$
      \end{itemize}
    \end{itemize}
  \end{enumerate}
\end{definition}
At first glance we might expect a condition requiring
$\rolvalop$ to respect inverse in the sense that
if $t_1$ and $t_2$ are inverses,
$\rolval{t_1}{r_1}$,
$\rolval{t_2}{r_2}$, then 
$\rtinv{r_1}{r_2} = true$.
But this is too much to ask given that \rolvalop is only a relation:
if the sort of $t$ is not \Akey or \Ikey then $t$ can have several 
images under $\rolvalop$, which will certainly will not be runtime
inverses of each other.

A valid transcript for a role \rl is a raw  transcript induced by a
valuation.

\begin{definition}[Valid Transcript for a Role]
  \label{def:valid-role-transcript}
  Let \rl be a role. %
  A transcript \tr is a 
  \emph{valid transcript for \rl} if 
  there exists a relation %
  $\rolvalop \subseteq \Term \times \rtval$
such that 
  \begin{align*}
   \text{ $\rolvalop$ is a valuation }
    \, \text{ and } \,
    (\mapR {\actmap{\rolvalop}} \ \rl \tr)
  \end{align*}
\end{definition}
We will sometimes drop the word ``valid,'' and sometimes say ``\tr is a
transcript for \pr'' if no confusion can arise.

\subsection{Transcripts for a Proc}
\label{sec:transcripts-proc}

As we did for roles, we want to define what it means for a raw transcript to
be a suitable transcript for a given \proc.
This is completely straightforward.

The first observation is that \procs have ``internal'' statements
(bindings and checks) that will not contribute to transcripts.
So we extract the statements we do care about:

\begin{definition}
  Let \pr be a \proc.  The body of \pr is a sequence of statements.  From
  this sequence we can extract the Events of the role and call this the 
  \emph{trace} of the proc.
\end{definition}

The trace of a \proc is a sequence of %
$\Act{\Loc}$ statements.   An execution of the \proc is given by an
assignment of runtime values to the locations of the \proc.

Similarly as for \roles we formalize the idea that a transcript \tr
``lines up'' with the actions of \proc \pr by saying that there exists
a relation %
$\sto : \Loc \to \rtval $ 
such that
\[
    (\mapR {\actmap{\sto}} \ (\trace\ \pr) \ \tr) \\
\]
But---just as for roles---we don't want to consider arbitrary
associations of values to locations. %
To say that the transcript really is a 
``run'' of the \proc is simply to say that it arises as the \proc
executes its bindings and checks.     So the relation between locations
and runtime values will be simply the store of the \proc as it executes.

Note the difference between the transcript notion for \procs and for
\roles: for \procs the relationship \sto between locations and runtime
values will be a \emph{function} not just a relation.  In a sense this
reflects the fact that \procs are deterministic programs.

So the correspondence condition induced by \sto on the actions of \pr
and \tr will simply be
\[
    \tr = \map {\actmap{\sto}} (\trace\ \pr)
\]
  Next we  record what it means to be a store for a \proc.

\subsubsection{Store for a Proc}
\label{sec:store-proc}
A store is a partial function from locations to runtime values that
respects the intended semantics of the bindings and checks.
A store \sto is a ``store for \pr'' if \sto respects the
statements of \pr in the obvious way.
\begin{definition}[Store]
  \label{def:store} %
  Let \pr be a \proc
  and let $\sto : \Loc \to \rtval $ 
  be a partial function on locations.
  Say that $\sto$ is a \emph{store for
    \pr}  if $\sto$ respects the checks and the (expressions of the)
  bindings of \pr in the following sense.
  
    \begin{enumerate}
    \item  
      if pr has $\CSrt{\loc}{s}$ %
      then $\rtsrt (\sto \loc) = s$. 

    \item 
      if pr has $\CSame{\loc_1}{ \loc_2}$ %
      then $\sto \loc_1 = \sto \loc_2$

    \item 
      if pr has $\CKypr {\loc_1} {\loc_2}$ %
      then $\sto \loc_1$ and $ \sto \loc_2$ make a runtime key pair,
      that is,
      $\rtpubof \; (\sto \loc_1) = (\sto \loc_2)$
      
    \item if \pr has $\CHash{\loc_1}{\loc_{2}}$
      then $\sto \loc_2 = \rthash (\sto \loc_1)$

    \item if \pr has $ \CQot{\loc}{s} $
      then $\sto \loc = \rtquote s$

    \item 
      if pr has
      $\Bind{t}{\loc}{\Pair{\loc_1}{\loc_2}}$
      then
      ${ \rtpair (\sto \loc_1) (\sto \loc_2) = (\sto v) }$
      holds

    \item 
      if pr has
      $\Bind{t}{\loc_e}{\Encr{\loc_p}{\loc_k}}$
      then
      ${ \rtencr (\sto \loc_p) (\sto \loc_k) (\sto \loc_e) }$
      holds

    \item  
      if pr has
      $ \Bind {t}{\loc}{\Hash{\loc_1}}$
      
      then
      ${(\sto \loc) =  \rthash (\sto \loc_1) }$
      holds.

    \item  
      if pr has
      $ \Bind {t}{\loc_1}{\PubOf{\loc}}$

      then
      $\rtpubof (\sto \loc)  \opteq {(\sto \loc_1)} $
      holds.

    \item 
      if pr has
      $ \Bind {t}{\loc_1}{\Frst{\loc}}$ 

      then
      ${\rtfrst (\sto v) \opteq (\sto \loc_1)  }$
      holds

    \item 
      if pr has
      $ \Bind {t}{\loc_1}{\Scnd{\loc}}$ 

      then
      ${\rtscnd (\sto v) \opteq (\sto \loc_1)  }$
      holds

    \item
      if pr has
      $ \Bind {t}{\loc_p}{\Decr{\loc_e}{\loc_k}}$ 

      then
      ${\rtdecr (\sto \loc_e) (\sto \loc_k) \opteq (\sto \loc_p) }$
      holds.

    \item
      if pr has
      $ \Bind {t} {\loc} {\Quot{s}}$ 

      then
      $(\sto \loc) = \rtquote{s}$
      holds.

    \end{enumerate}
\end{definition}

\begin{definition}[Valid Transcript for a Proc] 
  \label{def:valid-proc-transcript}
  Let \pr be a \proc. %
  A transcript \tr is a 
  \emph{valid transcript for \pr} if there exists a partial function %
  $\sto : \Loc \to \rtval$ such that 
  \begin{align*}
          \; \text{ $\sto$ is a store for \pr}
          \; \text{ and } 
    \tr = \map {\actmap{\sto}} (\trace \pr)
  \end{align*}
\end{definition}

\begin{example}   \label{eg:proc-store}
Let \pr be the following \proc.

  \begin{lstlisting}[frame=single,backgroundcolor=\color{yellow!10},
    caption={Proc for \refeg{eg:proc-store}}]
         Bind (Ch 1, L 1) (Param 1); 
         Csrt (L 1) Chan; 
         Evnt (Prm (L 2)); Bind (Nm 0, L 2) (Param 2);
         Csrt (L 2) Name; 
         Evnt (Prm (L 3));
         Bind (Ik (Av 2), L 3) (Param 3);
         Csrt (L 3) Ikey;

         Evnt (Rcv (L 1) (L 4));
         Bind (En (Nm 0) (Ak (Av 2)), L 4) (Read 1);
         Bind (Nm 0, L 5) (Decr (L 4) (L 3));
         Same (L 5) (L 2); 

         Evnt (Snd (L 1) (L 2))
\end{lstlisting}

Let the store \sto be defined as follows.
\begin{align*}
  \sto {\loc}_1 &= r_1 &
  \sto {\loc}_2 &= r_0 \\
  \sto {\loc}_3 &= r_2 &
  \sto {\loc}_4 &= r_3 \\
  \sto {\loc}_5 &= r_0 \\
\end{align*}
Then \sto is a store for \pr, as long as %
\begin{itemize}
\item $r_0$ is a value of sort \Nmop, [line 4]
\item $r_1$ is a channel value, [line 2]
\item $r_2$ is a value of sort \Ikey [line 7]
\item $r_0$ (the value at $\loc_5$) 
  is a decryption of $r_3$ (the value at $\loc_4$) by 
  $r_2$ (the value at $\loc_3$) [line 11 and line 12]
\end{itemize}
\qed \end{example}

\begin{example} \label{eg:proc-xpt}
Continuing \refeg{eg:proc-store}, the following transcript \tr
  \begin{align}
    &\Prm {r_1}; \ \Prm {r_2} ; \\
    &\Rcv {r_1} {r_3}; \\
    &\Snd {r_1} {r_4}
  \end{align}
is a valid transcript for our \proc.

To see why, first extract the trace of the \proc:

  \begin{lstlisting}[frame=single,backgroundcolor=\color{yellow!10},
    numbers=none,
caption={Trace of the \proc for \refeg{eg:proc-xpt}}]
         Evnt (Prm (L 2)); 
         Evnt (Prm (L 3));
         Evnt (Rcv (L 1) (L 4));
         Evnt (Snd (L 1) (L 2))
\end{lstlisting}

It is easy to see that 
    $\tr = \map {\actmap{\sto}} (\trace\ \pr)$
and we have seen in
\refeg{eg:proc-store} that it arises from the valid transcript \sto
\qed \end{example}

\section{Reflecting Transcripts}
\label{sec:relating}

If role \rl is compiled to \proc \pr, we want to know that any
execution of \pr is an execution of \rl.  This is captured
by the \emph{Reflecting Transcripts} theorem:

\begin{quote}
  \textbf{Theorem}
  Let \rl be a role, and suppose that \rl successfully compies to \proc
  \pr.  %
  Then any valid transcript for \pr is a valid transcript for \rl.
\end{quote}

\subsection{Outline of the Proof}
\label{sec:outline-proof}

Recall \refdef{def:tl}, relating terms to locations based on the
\Bindop statements in \proc:
\[
  \tl{t}{\loc} \eqdef 
  \; \text{ for some $e$,} \;
  (\Bind {t}{l} {e})  \in \pr
\]

Now suppose \tr is a transcript for \pr.  
By definition \tr is determined by a store function, from \Loc to
\rtval.   
By precomposing  this function with the relation \tlop , we get a
relation \tv from  \Term to \rtval,
as suggested by this picture.  
  \begin{center}
    \begin{tikzcd}
      \rl\arrow[rr, Rightarrow,  "\tlop"] 
      \arrow[dr,dashrightarrow, swap, "\rolvalop"]
      & & \pr \arrow[dl, "\sto"] \\
      & \tr
    \end{tikzcd}
  \end{center}  

  This yields (modulo lifting these functions and relations to the
  \Actop data type) a raw transcript.  It will not be hard to see that
  this transcript is in fact our original \tr, and is induced by \tv,
  and so \tr is a \emph{raw transcript for \rl}, induced by \tv.  To
  establish that \tr is a \emph{valid transcript for \rl} we need to
  show that it is induced by a valuation.  This is where the
  Saturation conditions on \pr come into play.

\subsection{The Proof}
\label{sec:proof}

\begin{theorem}[Reflecting Transcripts]
\label{thm:reflecting-transcripts}
  Let \rl be a role, and suppose that \rl successfully compies to \proc
  \pr.  %
  Then any transcript for \pr is a transcript for \rl.
\end{theorem}
\begin{proof}
  Let \tr be a transcript for \rl; %
  let \sto be a store for \rl that induces \tr.
  Viewing $\sto$ as a  relation, 
  let $\tv$ be the relational composition 
  \[ \tv \eqdef \tlop ; \sto
  \]
  That is
\[
    \tv \ t \ r \eqdef \exists\ l \ e, \Bind {t}{l} {e} \in \pr \land \
    \sto \ \loc \opteq r
\]
  To establish that \tr is a transcript for \pr it suffices to show
  that ${\tv}$ is a valuation that induces \tr.

  First:
  to show that ${\tv}$ induces \tr we want to show that 
\[
    \mapR  {\actmap{({\tv})}} \rl \ \tr
\]
  We  have 
\[
    \mapR \ ( \actmap{\tlop}) \ \rl \ (\trace \ \pr)
\]
  by definition of \tlop.

  We have
\[
    \mapR \ ({\actmap{\sto}}) \ (\trace \pr) \ \tr
\]
  since \tr is induced by \sto.

  By \reflem{lem:act-map-composition}, then, we have
  \begin{align*}
    \mapR\ (\actmap {(\sto ; \tlop) })\ \rl\ \pr ,
    &   \text{ that is, }
    \\
    \mapR\ (\actmap {\tv })\ \rl\ \pr 
  \end{align*}
  as desired.

  Next we want to show that ${\tv}$ is a valuation. %

  We consider the clauses of Definition~\ref{def:valuation} in turn.
  \begin{itemize}
  \item \rolvalop is functional on elementary terms:

    Given $t, r_1, r_2$ with  $t$ elementary,
    $\rolval{t}{r_1}$,  %
    and $\rolval{t}{r_2}$.
    We want to show $r_1 = r_2$.

    By definition of %
    $\rolval{t}{r_1}$  %
    and $\rolval{t}{r_2}$ %
    we have $\loc_1, e_1, \loc_2, e_2$ such that 
    \begin{align*}
      & \Bind{t}{\loc_1}{e_1} \text{ and } \sto (\loc_1) = r_1
      \\
      & \Bind{t}{\loc_2}{e_2} \text{ and } \sto (\loc_2) = r_2
    \end{align*}
    By the Check Equality Condition on \pr we have
    $\sameness{\loc_1}{\loc_{2}}$ in \pr. %
    Since \sto respects the sameness checks of \pr, we have
    $ \sto (\loc_1) = \sto (\loc_2)$ as desired.

  \item $\rolvalop$ respects sorts : 

    Given $t, r$ with  $t$ is elementary and 
    $\rolval{t}{r}$; we seek to establish that
    $\rtsrt \ r = \sortof t .$

    By definition of %
    $\rolval{t}{r}$  %
    we have $\loc$ and $e$ such that 
\[     \Bind{t}{\loc}{e} \text{ is in \pr and } \sto (\loc) = r
\]

    By the Check Sort Condition there is $\loc_1$ such that
    \begin{align*}
      \sameness{\loc}{\loc_1} \text{ and }
      \\
      \CSrt{\loc_1}{(\sortof t)}  \text{ in \pr.}
    \end{align*}
    Since \sto respects \CSrtop, 
    $$\rtsrt (\sto \loc_1) = \sortof t $$
    and since $\sameness{\loc}{\loc_1}$,
    $$\rtsrt r = \rtsrt (\sto \loc) = \sortof t$$
    as desired.

  \item  $\rolvalop$ respects pairing:

  Given $t_1, t_2, r$ %
  with $\rolval{\Prr{t_1}{t_2}}{r}$; %
  we seek  $r_1, r_2$ such that %
  $\rolval{t_1}{r_1}$, %
  $\rolval{t_2}{r_2}$, and %
  $\rtpair{r_1}{r_2} = r$.

  By definition of \tv we have  %
  $\loc$ and $e$ such that 
\[
    \Bind {\Prr{t_1}{t_2}} {\loc} {e} 
    \text{ is in \pr and } \sto (\loc) = r .
\]
  
  There are two cases: either $e$ is a pair expression for $\Pr{t_1}{t_2}$ or not.
  
  \begin{enumerate}
  \item If $e$ is a pair expression for $\Pr{t_1}{t_2}$ then
    we know that 
    $e$ is of the form $\Pair{\loc_1}{\loc_2}$  %
    and that there are $\loc_1, \loc_2, e_1,$ and $e_2$ such that 
    $\Bind{t_1}{\loc_1}{e_1}$ and $\Bind{t_2}{\loc_2}{e_2}$ are in \pr.

    Set $r_1$ to be $\sto \loc_1$ and
    $r_2$ to be $\sto \loc_2$.

    Then
    \[      \rolval {t_1}{r_1}
      \quad \text{ and } \quad
      \rolval {t_2}{r_2} .
    \]
    It remains to show
    $\rtpair {r_1}{r_2} = {r}$
    
    Since \sto respects \Prrop and pr has
    \[
      \Bind {\Prr{t_1}{t_2}} {\loc} {\Pair {\loc_1} {\loc_2} }
    \]
    we have 
    \[ { \rtpair\  (\sto \loc_1)\ (\sto \loc_2) = (\sto \loc) }
    \]
    which is to say
    \[      { \rtpair\ r_1 \; r_2 \; r }
    \]
    as desired.

  \item If $e$ is not a Pair expression for $\Pr{t_1}{t_2}$ then %
    we use the fact that the runtime satisfies
    \[
      \rtpair\ r_1 \ r_2 = r 
      \ \leftrightarrow \
        \rtfrst r = r_1 \land \rtscnd r = r_2 
      \]
    and so we exhibit $r_1$ and $r_2$ with  %
    $  \rtfrst r = r_1 \land \rtscnd r = r_2 . $

    Since $e$ is not 
    a pair expression for $\Pr{t_1}{t_2}$,
    the Pair Elimination Conditions hold. %
    So we have $\loc_1$ and $\loc_2$ such that
    \begin{align*}
      \Bind {t_1}{\loc_1}{\Frst{\loc}} \\
      \Bind {t_2}{\loc_2}{\Scnd{\loc}}
    \end{align*}

    Set $r_1 = \sto \loc_1$ and $r_2 = \sto \loc_2$
  \end{enumerate}
  Since \sto respects \Frstop  and \Scndop,
  \begin{align*}
    \rtfrst (\sto \loc) \opteq (\sto \loc_1) 
    &\text{ and }
    \rtscnd (\sto \loc) \opteq (\sto \loc_2) 
  \end{align*}
  that is, 
  \begin{align*}
    \rtfrst r  \opteq r_1
    &\text{ and }
    \rtscnd r \opteq r_2
  \end{align*}
   as desired.

  \item  $\rolvalop$ respects hashing:

    Given $t, r$ %
    with $\rolval{\Hs{t}}{r}$; we 
    seek
     $r_t$ such that
    $\rolval{t}{r_t}$ and
    $\rthash{r_t} = r$.

    By definition of \tv we have %
    $\loc_h, e_h, $ such that 
    \begin{align*}
      \Bind{t}{\loc_h}{e_h} \text{ and } \sto (\loc_h) = r
    \end{align*}

    By the Hash Justified Condition on \pr there are %
    $\loc_t$ and $e_t$ with
    \[ \Bind{t}{\loc_t}{e_t} \text{ in } \pr
    \]

    Then by the Check Hash Condition on \pr
    \begin{align*}
      \CHash {\loc_h}{\loc_t} \text{ is in } \pr
    \end{align*}
    Take $r_t$ to be $\sto r_t$; since \sto respects \CHashop, 
    \begin{align*}
      \sto \loc_h &= \rthash (\sto \loc_t) , 
      \quad \text{ that is, }    \\
     r_h &= \rthash r_t
    \end{align*}

  \item $\rolvalop$ respects quote:

    Given $s, r$ %
    with $\rolval{\Qt{s}}{r}$;
    we want to show $ r = \rtquote s$

    By definition of \tv we have  $\loc, e$ such that 
\[ \Bind{\Qt{s}}{\loc}{e} \text{ is in \pr and } \sto (\loc) = r
\]
    By the Check Quote Condition
    \[ \CQot{\loc}{s} \text{ is in \pr } \]
    Since \sto respects $\CQotop$, %
    $\sto \loc = \rtquote s$

  \item $\rolvalop$ respects key pairs

    Given $t_1, t_2, r_1, r_2$, %
    with $t_1$ and $t_2$ making a symbolic key pair and
    $\rolval{t_1}{r_1}$,
    $\rolval{t_2}{r_2}$; %
    we want to show
    $\rtkypr \ {r_1}\ {r_2} = true. $

    By definition of \tv we have  
    $\loc_1, e_1, \loc_2, e_2$ such that
    \begin{align*}
       \Bind{t_1}{\loc_1}{e_1} \text{ is in \pr and } \sto (\loc_1) = r_1
      \\
       \Bind{t_2}{\loc_2}{e_2} \text{ is in \pr and } \sto (\loc_2) = r_2
    \end{align*}

    By the Check Key Pair Condition
    \begin{align*}
      \exists \loc'_1 \; \loc'_2 \; e'_1 \; e'_2 ,
        \Bind{t_1}{\loc'_1}{e'_1} 
        \land\
        \Bind{t_2}{\loc'_2}{e'_2} \
        \land\
      \\
        \sameness{\loc_1}{\loc'_{1}},
        \land\
        \sameness{\loc_2}{\loc'_{2}}
        \land\
        \CKypr{\loc'_1}{\loc'_2}
    \end{align*}
    
        Let $r'_1 = \sto r_1$
        and $r'_2 = \sto r_2$.

        Since $\sameness{\loc_1}{\loc'_{1}},
        \land\ 
        \sameness{\loc_2}{\loc'_{2}}$, and the fact that \sto respects
        $\samenessop$ it suffices to show that 
        $\rtkypr \ {r'_1}\ {r'_2} = true. $

        But that follows from the facts that 
        $\CKypr{\loc'_1}{\loc'_2}$ is in \pr and \sto 
        respects \CKyprop

\item $\rolvalop$ respects encryption:

  Given $p, k_e, r_e$, %
  with $\rolval {\En{p}{k_e}} {r_e}$; %
  we want to establish the disjunctive encryption property.

  By definition of \tv we have  %
  $\loc$ and $e$ such that 
  \begin{align}
    \Bind {\En{p}{k_e}} {\loc} {e} &\in \pr \label{encr-given}
    \\
    \sto (\loc) &= r_e .  \label{sto-loc}
  \end{align}
  There are two cases: either $e$ is 
  an encryption expression for $\En{p}{k_e}$
  or not.
  
  \begin{enumerate}
  \item Suppose $e$ is   an encryption expression for $\En{p}{k_e}$,
    so that
    \begin{align} \label{encr-exp}
      \Bind {\En{p}{k_e}} {\loc_e} {\Encr {\loc_p} {\loc_k} } 
      &\in \pr
    \end{align}
    for some $\loc_p$ and $\loc_k$.
   We establish the encr condition, \ie, that 
        there exists $r_p, r_{ke}$ such that 
    \begin{itemize}
    \item $\rolval {p}{r_p}$
    \item $\rolval {k_e}{r_{ke}}$
    \item $\rtencr \; {r_p} \; {r_k} \;{r_e}$
    \end{itemize}

    Since $e$ is   an encryption expression for $\En{p}{k_e}$,
    there are $ e_p$ and $e_k$
    such that 
    \begin{align}
      \Bind{t_p}{\loc_p}{e_p} &\in \pr
      \\
      \Bind{t_k}{\loc_k}{e_k} &\in \pr
    \end{align}
    
    Set $r_p$ to be $(\sto \loc_p)$ and
    $r_{ke}$ to be $(\sto \loc_k)$.
    Then
    \begin{align}
      \rolval {p}{r_p} \label{eq:tvprp}
      \\
      \rolval {k_e}{r_{ke}} . \label{eq:tvkerke}
    \end{align}

    Since \sto respects \Encrop we have, by (\ref{encr-exp}),
    \[ { \rtencr (\sto \loc_p) (\sto \loc_k) (\sto \loc) }
    \]
    which is to say
    \begin{align} \label{eq:rtencr}
      { \rtencr r_p \; r_k \; r_e }.
    \end{align}

 The encr condition follows from 
    (\ref{eq:tvprp}), (\ref{eq:tvkerke}), and (\ref{eq:rtencr})

  \item Suppose $e$ is not 
  an encryption expression for $\En{p}{k_e}$.
    We establish the decr condition, \ie, that %
    there exist $r_p$ and $r_{kd}$ such that 
    \begin{itemize}
    \item $\rolval {p}{r_p}$
    \item $\rolval {\inv{k_e} } {r_{kd}}$
    \item $ \rtdecr\ r_e\ r_{kd} \opteq (\sto v_{p})$
    \end{itemize}

    By the {Encryption Justification} property applied to
    (\ref{encr-given}) there are
    $\loc_{kd}$ and $e_{kd}$ such that 
    \begin{align}
      \Bind {\inv{(k_e)}}{\loc_{kd}}{e_{kd}} &\in \pr
    \end{align}
    Set $r_{kd}$ to be $(\sto \loc_{kd})$, thus
    \begin{align} \label{eq:tvinvke}
      \tv {\inv{(k_e)}} {r_{kd}}
    \end{align}

    By the  {Decryption Condition} applied to
    (\ref{encr-given}) and (\ref{sto-loc}) we have
    \begin{align} \label{eq:decr}
      \Bind{p}{\loc_p}{\Decr{\loc}{\loc_{kd}}} &\in \pr
    \end{align}
    for some $\loc_p.$
    Set $r_{p}$ to be $(\sto \loc_{p})$, thus
    \begin{align} \label{eq:tvprp'}
      \tv \ {p} \ {r_p}
    \end{align}
    
    Since \sto respects \Decrop, we have, by (\ref{eq:decr}),
    \begin{align*}
      \rtdecr\ (\sto \loc) (\sto v_{kd}) \opteq (\sto v_{p})
    \end{align*}
    which is to say
    \begin{align}\label{eq:rtdecr}
      \rtdecr\ r_e\ r_{kd} \opteq r_p
    \end{align}

    The decryption condition follows from %
    (\ref{eq:tvinvke}), (\ref{eq:tvprp'}), and (\ref{eq:rtdecr})

  \end{enumerate}
\end{itemize}
\end{proof}

\begin{remark}
  It is instructive to 
  compare the treatments of   pairing and encryption %
  in the above proof.  
  We start with
  \[
    \Bind {\Prr {t_1}{t_2}} {\loc} {e}
    \; \text{ or } \; 
    \Bind {\En {p}{k}} {\loc} {e}
  \]
  In each instance the
  argument branched on whether or not the expression $e$ was a Pair
  expression, or Encryption expression, respectively.  In the
  affirmative case for each instance we argued directly that \tv
  satisfied the definition of valuation, and the arguments were
  precisely parallel.
  
  In the neutral cases
  we argued indirectly:
  \begin{itemize}
  \item  using the axiom %
    \[ \rtpair\ r_1 \ r_2 = r  \ \leftrightarrow \
      \rtfrst r = r_1 \land \rtscnd r = r_2 
    \]
    for pairing, and
  \item using the ``decr condition'' for encryption.
  \end{itemize}
  The pairing case was simpler. %
  Why, for the encryption case, did we not just invoke the equivalence similar
  to that for pairing, namely
  \[
    \rtencr\ r_p \ r_{k} \ r_e \ \leftrightarrow \
    \rtdecr \ r_e \ \inv{r_k} = r_p 
  \]
  instead of going to the trouble of defining
  the decr condition?   Here's the explanation.

  In the encryption proof, the runtime value $r$ is the value of the store \sto on
  the location for the key $k$. The equivalence
  about encryption refers to both $r_k$ and $\inv{r_k}$.  The latter
  would arise naturally as the value of \sto on $\inv{k}$. 
  But there is
  no reason to suppose that our \proc \pr has locations corresponding to
  each of $k$ and $\inv{k}$.  The situation for pairing is simpler in that
  the pairing equivalence involves no ``alien'' value analogous to
  $\inv{r_k}$.
  Logically speaking, deconstructing an encryption involves a minor
  premise, not so for  deconstructing a pair.

  In essence our proof in the encryption case is branching on whether
  the \proc has a binding for $k$ or a binding for $\inv{k}$; in the
  latter case we are using the fact that in a saturated \proc an
  encryption binding with a neutral expression is \emph{justified}.

  This (natural!) inconvenience that \pr probably does not
  have locations corresponding to each of $k$ and its inverse is
  precisely why the definition of valuation has its disjunctive
  character.

  Section~\ref{sec:strong-valuation} outlines an alternative approach
  to the semantics and our proofs that is related to this issue.
\end{remark}

\section{Discussion}

\subsection{Preserving Transcripts}
\label{sec:preserving-transcripts}

It is natural to consider a converse to Reflecting Transcripts, namely
that if role \rl is compiled to \proc \pr, then any activity consistent
with \rl is a possible execution of \pr.  Formally, we might seek a
result claiming:
\begin{quote}\em
  If \rl compiles to \pr then any transcript for \rl is a transcript for \pr.
\end{quote}
But in fact a ``preserving transcripts'' in this form does not hold for \molly.

\begin{example}
  \label{eg:no-preserving}
  Suppose $ch$ is a channel term, $k$ is a symmetric key, and $p$ is
  arbitrary.
  Consider the artificially simple role containing two transmissions of the 
  randomized encryption $\En{p}{k}$.
 \begin{align*}
    & \Prm{ch} ; \Prm{p} ; \Prm{k} \\          
    & \Snd{ch}{\En{p}{k}} \\
    & \Snd{ch}{\En{p}{k}} 
  \end{align*}
Now consider a term valuation \rolvalop in which
the term $\En{p}{k}$ is associated with 
a non-singleton set $R$ of values.
Then \rolvalop supports many transcripts in which 
the last two values sent are different values in $R.$

But any compiler that allocates only one location $\loc$ for 
the term $\En{p}{k}$ cannot have a valid transcript with different values in its last 2
places.   If \sto is a store giving rise to a transcript for  \proc, the
last two values must both be $\sto(\loc)$.
\qed \end{example}

This little example shows that any compiler supporting a Preserving
Transcripts theorem as stated above must, at least, allocate as many
locations to each role-term as there are occurrences of that term.
This seems artificial, especially since the role specification
language we have worked with does not provide any support for
ascribing different behaviors with different occurrences of a given
term.

\subsection{Strong Term Valuations}
\label{sec:strong-valuation}
This section should shed some light on the notion of valuation for a
role, specifically the definition of a valuation ``respecting encryption.''

Here is another natural definition of valuation for a term.
It replaces the disjunctive character of our official definition of
respecting encryption, and strengthens the requirement concerning key
pairs (compare Definition~\ref{def:valuation}).

\begin{definition}[Strong Valuation]
  \label{def:strong-valuation}
  A relation   $\rolvalop \subseteq \Term \times \rtval$
  is a \emph{strong term valuation} if it satisfies the conditions in
  \refdef{def:valuation} with the following two strengthening of the
  conditions there about encryption and key pairs.

\begin{itemize}

 \item  \rolvalop strongly respects key pairs:

    For all $t_1, t_2, r_1$, %
    \begin{itemize}
    \item 
    if $(t_1, t_2) $ makes a key pair and
    $\rolval{t_1}{r_1}$, then
    \begin{itemize}
    \item there exists $r_2$ with 
      $\rolval{t_2}{r_2}$,  and
    \item for all  $r_2$ with 
      $\rolval{t_2}{r_2}$,
    $\rtkypr{r_1}{r_2} = true$
    \end{itemize}
    \item 
    if $(t_2, t_1) $ makes a key pair and
    $\rolval{t_2}{r_2}$, then
    \begin{itemize}
    \item there exists $r_1$ with 
      $\rolval{t_1}{r_1}$,  and
    \item for all  $r_1$ with 
      $\rolval{t_1}{r_1}$,
      $\rtkypr{r_2}{r_1} = true$
    \end{itemize}
  \end{itemize}

   \item \rolvalop strongly respects encryption:

    For all $p, k_e, r_e$, %
    if $\rolval {\En{p}{k_e}} {r_e}$ then 
    there exists $r_p, r_{ke}$ such that 
    \begin{itemize}
    \item $\rolval {k_e}{r_{ke}}$
    \item $\rolval {p}{r_p}$
    \item $\rtencr {r_p}{r_{ke}}{r_e}$
    \end{itemize}
  \end{itemize}
\end{definition}

Clearly if \rolvalop is a strong valuation then it is a valuation.
What we show in this section is that if \rolvalop is a valuation that
fails to be a strong valuation there is a natural  way to extend it to a
strong valuation.

Recall that when $r$ is a runtime value we use $\inv{r}$ to refer to
the unique value $r'$ acting as runtime inverse of $r$
(\refdef{def:rt-inverse} ), even though $r'$ might not be feasibly
computable from $r$.
\begin{definition} \label{def:completion}
  Let \rolvalop be a term valuation. %
  Define the \emph{completion} $\completion{\rolvalop}$ of \rolvalop to be%
  \[
    \rolvalop \; \cup \; \{ (\inv{t}, \inv{r}) \mid (t,r) \in \rolvalop \}
  \]
\end{definition}

\begin{lemma} \label{lem:pre-val-val}
  If \rolvalop is a term valuation, %
  then $\widehat{\rolvalop}$ is a strong term valuation.
\end{lemma}
\begin{proof}
  Since the only difference between \rolvalop and 
  $\widehat{\rolvalop}$ %
  is the addition of some pairs whose term is
  elementary, 
  $\widehat{\rolvalop}$  certainly respects pairing, hashing and quotes.
  It is easy to see that 
  $\widehat{\rolvalop}$  is functional on elementary terms and respects
  sorts.
  The strong key pair condition holds for 
  $\widehat{\rolvalop}$ by construction.

  It remains to show that 
  $\widehat{\rolvalop}$  respects encryption in the sense of 
  Definition~\ref{def:strong-valuation}.

  So choose %
  $p, k_e, r_e$ such that %
  $\rolval {\En{p}{k_e}} {r_e}$.  Since $\widehat{\rolvalop}$ satisfies
  the decr condition we have $r_p, r_{kd}$ such that
  \begin{itemize}
  \item $\rolval {\inv{(k_e)}}{r_{kd}}$
  \item $\rolval {p}{r_p}$
  \item $\rtdecr {r_e}\ {r_{kd}} = Some\ {r_p}$ .
  \end{itemize}
  We seek
  $r_p, r_{ke}$ such that 
  \begin{itemize}
  \item $\rolval {k_e}{r_{ke}}$
  \item $\rolval {p}{r_p}$
  \item $\rtencr\ {r_p}\ {r_k}\ {r_e}$ .
  \end{itemize}
  Of course the $r_p$ we seek is the given $r_p$; then take
  $r_{ke}$ to be $\inv{r_{kd}}$.  We have %
  $\rolval {\inv{(k_e)}}{r_{kd}}$  by the construction of 
  $\widehat{\rolvalop}$, and
  $(\rtencr\ {r_p}\ {r_k}\ {r_e})$ holds by the axiom
  \[
    \rtencr\ r_p \ r_{k} \ r_e \ \leftrightarrow \
    \rtdecr \ r_e \ \inv{r_k} = r_p 
\]
\end{proof}

Lemma~\ref{lem:pre-val-val} shows that we could have worked with
strong valuations all along.

Strong valuation is perhaps more natural as a formalization of the
meanings of symbolic terms in a role.  But valuations as we originally
defined them are the right definition for working with the data
\emph{explicitly available} from a given role: we typically only
``have'' one half of a key pair.

If we had used the notion of strong valuation when defining transcripts
the main change in our development would
come in the proof of the Reflecting Transcripts theorem. %
There we had
to construct a valuation \tv.  This valuation was defined quite
simply from the \proc and the store as a composition.  If we had
needed to build a strong valuation we'd have built our (strong) valuation as
the completion of this composition.  It would have worked out but all
the arguments would have been a bit more clumsy since we'd have had to
take the completion into account at every step in the argument.

\section{Future Work}
\label{sec:future-work}

{\molly} currently handles just a minimum set of primitives to
exercise the relevant algorithm ideas and proof techniques.   We expect
that it will be straightforward to expand to other operations such as
signatures and richer notions of tupling.   Extensions to the message
algebra that have interesting \emph{semantic} consequences, such as
Diffie-Hellman operations or  rich equational theories, will demand
some care but, we expect, no changes to our basic approach.

A more significant extension of the current work will connect it with
protocol \emph{analysis}.  Specifically we plan to integrate {\molly}
into the \cpsa\ ecosystem, in a way such that a protocol designer,
having established some security goals for her protocol using a \cpsa\
analysis, can generate implementations of the various
protocol roles, obtaining a version satisfying those goals.  The main challenge
here lies in the fact that our current correctness claims, based on
transcripts, relate the symbolic and runtime traces of individual
roles.  To draw conclusion about the runtime semantics of a full
{protocol} execution requires attention to the interactions among
transcripts.
This is a distributed activity involving a number of different
compliant participants as well as possibly adversarial actions.  We
will turn to that larger question in a subsequent paper, using the
current results as a basis.
